\documentclass[journal]{IEEEtran}

\usepackage[usenames]{color}
\usepackage{times}
\usepackage{algorithmic}
\usepackage{amsmath}
\usepackage{amssymb}
\usepackage{amsthm}
\usepackage{dsfont}
\usepackage{color}
\usepackage{colortbl}
\usepackage{float}
\usepackage{cite}
\usepackage{verbatim}
\usepackage[pdftex]{graphicx}
\usepackage{caption}
\usepackage{subcaption}

\newtheorem{theorem}{Theorem}
\newtheorem{proposition}[theorem]{Proposition}
\newtheorem{corollary}[theorem]{Corollary}
\newtheorem{example}[theorem]{Example}
\newtheorem{lemma}[theorem]{Lemma}
\newtheorem{definition}[theorem]{Definition}
\newtheorem{remark}[theorem]{Remark}

\newcommand{\beq}{\begin{eqnarray}}
\newcommand{\eeq}{\end{eqnarray}}
\newcommand{\n}{\nonumber}

\newcommand{\field}[1]{\mathbb{#1}}

\newcommand{\F}{\field{F}}

\newcommand{\cM}{{\cal M}}
\newcommand{\cC}{{\cal C}}

\newfont{\bbb}{msbm10 scaled 500}

\newfont{\bb}{msbm10 scaled 1100}

\newcommand{\av}{{\bf a}}
\newcommand{\bv}{{\bf b}}
\newcommand{\cv}{{\bf c}}
\newcommand{\dv}{{\bf d}}
\newcommand{\ev}{{\bf e}}
\newcommand{\fv}{{\bf f}}
\newcommand{\gv}{{\bf g}}

\newcommand{\lv}{{\bf l}}

\newcommand{\pv}{{\bf p}}

\newcommand{\rv}{{\bf r}}
\newcommand{\sv}{{\bf s}}

\newcommand{\uv}{{\bf u}}

\newcommand{\vv}{{\bf v}}
\newcommand{\xv}{{\bf x}}
\newcommand{\yv}{{\bf y}}

\newcommand{\Um}{{\bf U}}

\newcommand{\Vm}{{\bf V}}

\newcommand{\Ac}{{\cal A}}
\newcommand{\Bc}{{\cal B}}
\newcommand{\Cc}{{\cal C}}
\newcommand{\Dc}{{\cal D}}
\newcommand{\Ec}{{\cal E}}

\newcommand{\Gc}{{\cal G}}

\newcommand{\Kc}{{\cal K}}

\newcommand{\Mc}{{\cal M}}

\newcommand{\Rc}{{\cal R}}
\newcommand{\Sc}{{\cal S}}

\newcommand{\Xc}{{\cal X}}

\newcommand{\nuv}{\hbox{\boldmath$\nu$}}

\definecolor{OXO-emph}{RGB}{153,0,0}

\newcommand\ceilb[1]{\left\lceil #1 \right\rceil}
\newcommand\floorb[1]{\left\lfloor #1 \right\rfloor}
\newcommand{\algrule}[1][.2pt]{\par\vskip.5\baselineskip\hrule height #1\par\vskip.5\baselineskip}
\DeclareMathAlphabet{\mathpzc}{OT1}{pzc}{m}{it}

\begin{document}

\title{Optimal Locally Repairable and Secure Codes\\ for Distributed Storage Systems}

\author{Ankit~Singh~Rawat, O.~Ozan~Koyluoglu, Natalia~Silberstein, and~Sriram~Vishwanath
\thanks{The authors are with the Laboratory of Informatics, Networks and Communications, Department of Electrical and Computer Engineering, The University of Texas at Austin, Austin, TX 78751 USA. E-mail: ankitsr@utexas.edu, \{ozan, natalys, sriram\}@austin.utexas.edu.}
\thanks{This paper was presented in parts at the
Workshop on Trends in Coding Theory, Ascona, Switzerland, October 2012; at IEEE Information Theory and Applications Workshop, San Diego, CA, February 2013; and at IEEE International Symposium on Information Theory, Istanbul, Turkey, July 2013.}}

\maketitle

\begin{abstract}
This paper aims to go beyond resilience into the study of security and local-repairability for distributed storage systems (DSS). Security and local-repairability are both important as features of an efficient storage system, and this paper aims to understand the trade-offs between resilience, security, and  local-repairability in these systems. In particular, this paper first investigates security in the presence of colluding eavesdroppers, where eavesdroppers are assumed to work together in decoding stored information. Second, the paper focuses on coding schemes that enable optimal local repairs. It further brings these two concepts together, to develop locally repairable coding schemes for DSS that are secure against eavesdroppers.

The main results of this paper include: a. An improved bound on the secrecy capacity for minimum storage regenerating codes, b. secure coding schemes that achieve the bound for some special cases,  c. a new bound on minimum distance  for locally repairable codes, d.  code  construction for locally repairable codes that attain the minimum distance bound, and e. repair-bandwidth-efficient locally repairable codes with and without security constraints.
\end{abstract}


\IEEEpeerreviewmaketitle


\section{Introduction}

\subsection{Background}

Distributed storage systems (DSS) are of increasingly importance, given the vast amounts of data being generated and accessed worldwide. OceanStore \cite{oceanstore}, Google File System (GFS) \cite{gfs} and TotalRecall \cite{totalrecall} are a few examples of existing DSS. An essential component of DSS is resilience to node failures, which is why every DSS today incorporates a mechanism for protection against failures, thus preventing permanent loss of  data stored on the system. Typically, this resilience is afforded by replication, and in recent years, using more sophisticated coding theoretic techniques.

Node failures are one of the many design challenges faced by DSS. There are two other challenges, arguably of equal importance: 1) security and 2) local-repairability. Due to the decentralized nature of such systems, it is important whether they are secure against a variety of possible attacks. Our focus in this paper is on passive eavesdroppers located at multiple nodes in the DSS that can collude in attempting to gain an understanding of the stored data. In addition to being decentralized, DSS systems are often widely geographically distributed; therefore, local-repairability in storage proves very useful. In this paper, we develop a deeper understanding of local-repairability in storage, and subsequently combine locality and security to develop codes for secure locally-repairable DSS.

The security of communication or storage systems can be analyzed with their resilience to active or passive attacks~\cite{Goldreich:Foundations04, Delfs:Introducion07}. Active attacks in such systems include settings were an adversary modifies existing packets or injects new ones into the system, whereas the passive attack models include eavesdroppers observing the messages being stored/transmitted. For DSS, cryptographic approaches are often ineffective, as key distribution and management between all nodes in the system is extremely challenging to accomplish. A coding/information theoretic approach to security typically offers stronger security guarantees than cryptographic schemes \cite{Shannon:Communication49,Wyner:The75}. Moreover, in the context of distributed storage, such approach is logistically easier to realize than mechanisms that require key management. A secrecy-enabling coding scheme is designed based on a worst-case estimate of the information leaked to eavesdroppers, and can naturally complement other existing coding schemes being utilized in distributed storage systems. In its simplest form, security against an eavesdropper can be achieved using a one-time pad scheme~\cite{Vernam:Cipher26}. For example, consider that the contents of the two nodes are given by $X_1=R$, and $X_2=R \oplus D$, where $R$ is a uniformly random bit, and $D$ is the data bit. By contacting both nodes, one can clearly obtain the data by computing $X_1\oplus X_2$. However, one can not get any information about the data bit by observing any one of the two nodes as the mutual information between the data and the content of one of the nodes is zero, i.e., $I(X_i;D)=0$ for $i=1,2$. This example corroborates the significance of information theoretic approach in securing DSS. 

Local-repairability of DSS is an additional property which can be one of the primary design criteria for the system. The corresponding performance metric associated with a coding scheme is its  {\em locality} $r$ which is defined as the number of  nodes that must participate  in a  repair process when a particular node fails.  Local-repairability requires small locality, which implies that fewer nodes are involved in the node repair process. This makes the entire repair process easier from a logistical perspective. In addition, local-repairability is of significant interest when a cost is associated with contacting each node in the system. Local-repairability, in its simplest form, can be accomplished by splitting the data into groups, and then separately coding and storing each group. However, this na\"{\i}ve approach requires the connection to all the groups in order to retrieve the whole data, and may not be the most efficient in terms of other metrics of interest, e.g., resilience. Therefore, there is a growing interest in more sophisticated mechanisms for achieving local-repairability in DSS. The systems designed with locality in mind can also present benefits in terms of security. In other words, local-repairability and security against eavesdropper attack go hand in hand, and, as we illustrate in this paper, a joint design of both features can prove to be particularly useful.

In DSS, encoding data before storing it on the system provides the same level of resilience against node failures as that of the conventional approach of uncoded replication, but with much less storage space. The advantages that can be leveraged in terms of storage space may result in a degradation of other performance metrics. Being one of such metrics, {\em repair bandwidth} refers to the amount of data that needs to be transferred in the event of a single node failure in order to regenerate the data on the failed node. This metric is highly relevant as a large fraction of network bandwidth in DSS can be occupied by the data being transferred during the repair process. Thus, it is desirable to have coding schemes with small repair bandwidth. Most of the codes designed for storage in the past are {\em maximum distance separable} (MDS) codes~\cite{BlaBraBruJai, BlaRoth, CasBruck}. These codes usually entail a high repair bandwidth as the entire original file is reconstructed in order to regenerate the encoded data stored at a particular storage node. The DSS employing MDS codes have {\em `any $k$ out of $n$' property}, where the content of any $k$ storage nodes out of $n$ nodes in the system is sufficient to recover the original file. 

In \cite{dimakis}, Dimakis et al. address the issue of repair bandwidth and establish a trade-off between per node storage $\alpha$ and repair bandwidth $\gamma$ for codes that have desirable `any $k$ out of $n$' property\footnote{These codes are not necessarily MDS codes. For MDS codes, we have the file size $\Mc = k\alpha$, in addition to `any $k$ out of $n$' property.}. The class of codes that attain this trade-off  is referred to as \emph{regenerating codes}. Utilizing network coding framework, \cite{dimakis} considers the notion of {\em functional repair} where node repair process may not exactly replicate  the original failed node; nonetheless, the repaired node preserves `any $k$ out of $n$' property of the overall storage system. However, it is desirable to perform {\em exact repair} in DSS, where the data regenerated after the repair process is an exact replica of what was stored on the failed node. This is essential due to the ease of maintenance and other practical advantages, e.g., maintaining a code in its systematic form. Exact repair is also advantageous compared to the functional repair in the presence of eavesdroppers, as the latter scheme requires updating the coding rules which may leak additional information to eavesdroppers \cite{pawar11}. Noting the resilience of exact repair to eavesdropping attacks and the necessity of it for practical purposes, it is of significant interest to design regenerating codes that not only enjoy an optimal trade-off in repair bandwidth vs. storage, but also satisfy exact repair in addition to security and/or small locality constraints.


\subsection{Contributions and organization}

In this paper, we consider secure minimum storage regenerating codes, and locally repairable codes for DSS with/without security constraints. As a security constraint, we adopt the passive and colluding eavesdropper model presented in \cite{SRK_globecom11}. Under this model, during the entire life span of the DSS, the eavesdropper can get access to data stored on an $\ell_1$ number of nodes, and, in addition, it observes both the stored content and the data downloaded (for repair) on an additional $\ell_2$ number of nodes. This attack model generalizes the eavesdropper model proposed in \cite{pawar11}, which considers the case of $\ell_2=0$. Since the amount of information downloaded during a node repair is equal to the information stored on the repaired node for minimum bandwidth regenerating codes, the two notions are different only for the minimum storage regenerating codes.

With this general eavesdropper model, we extend the existing results on the design of secure minimum storage regenerating codes for DSS in Section~\ref{sec:secrecy_BW}. First, we derive an upper bound on secrecy capacity, the amount of data that can be stored on the system without leaking information to an eavesdropper, for a DSS employing bandwidth efficient node repair. Our bound is novel in that it can take into account the additional downloaded data at the eavesdroppers; as a result, it is tighter than the available bounds in the literature~\cite{pawar11}. In Section~\ref{subsec:secure_msr}, we present a secure, exact-repairable coding scheme that has a higher code rate compared to that of the scheme presented in~\cite{SRK_globecom11}. Utilizing a special case of the obtained bound, we show our coding scheme achieves the bound on secrecy capacity  for any $(\ell_1,\ell_2)$ when the eavesdropper observes node repairs of a set of $\ell_1+\ell_2$ ($<k$) nodes for any $\ell_2\leq 2$.

In Section~\ref{sec:lrc}, we consider locally repairable DSS where an original file is encoded using a code with (small) locality before storing it on the system. In particular, we focus on vector codes with \emph{all symbol locality} which we refer to as vector locally repairable codes. Similar to scalar codes with locality considered in \cite{PKLK12}, we allow for these codes to have multiple local parities per local group. In Section~\ref{subsec:d_min_upper}, we derive an upper bound on minimum distance of vector locally repairable codes. The bound is quite general and applies to non-linear codes as well. We develop this bound using the proof technique used in \cite{Gopalan12}, \cite{DimDim12}. Subsequently, we present a novel explicit construction for vector locally  repairable codes, based on Gabidulin codes, a family of {\em maximum rank distance} (MRD) codes. We show that these codes are optimal in the sense that they achieve the derived bound on minimum distance. Remarkably, this establishes a per node storage vs. resilience trade-off for vector locally repairable codes with multiple local parities per local group, similar to the work of  Papailiopoulos et al~\cite{DimDim12}, which considers vector locally repairable codes with a single local parity per local group. Here, as a special case of our general construction for vector locally repairable codes, we also describe an explicit construction for scalar code with the all symbol locality property for the range of parameters where only existential results are present in the literature~\cite{Gopalan12, PKLK12}.

In Section~\ref{sec:local_repair_bw}, we introduce the notion of repair bandwidth efficiency for locally repairable DSS. We obtain an upper bound on the file size  for locally repairable codes that allow node repairs to be performed with a given repair bandwidth. We then present minimum distance optimal locally repairable coding schemes that are repair bandwidth efficient and achieve the file size bound\footnote{In a parallel and independent work~\cite{KPLV12}, Kamath et al. also provide upper bounds on minimum distance together with constructions and existence results for vector locally repairable codes. \cite{KPLV12} studies repair bandwidth efficient locally repairable codes as well.}. We combine locally repairable codes with minimum storage regenerating codes to obtain such schemes.

Finally, we consider the problem of providing security against passive eavesdropping attacks for locally repairable codes in Section~\ref{sec:secrecy_local}. We derive an upper bound on the size of data that can be stored on a minimum distance optimal locally repairable DSS that is secure against an $(\ell_1, \ell_2)$-eavesdropper. We consider two cases: (i) single parity node per local group in Section~\ref{subsec:single_local} and (ii) multiple parity nodes per local group in Section~\ref{subsec:secure_LRC_2}. We present secure locally repairable coding schemes for DSS in both cases using secrecy pre-coding. We also highlight the advantages of having multiple local parity nodes per local group when security constraints are present.

In all the scenarios that we study in this paper, the achievability results allow for \emph{exact} repair, and we obtain secure file size upper bounds  from min-cut analyses over the secrecy graph representation of DSS. Our main secrecy achievability coding argument is obtained by utilizing a secret sharing scheme with MRD codes (based on linearized polynomials), similar to the classical work of~\cite{shamir79}. In the following subsection, we present a summary of related work to the problems studied in this paper. In Section II, we provide a general system model together with necessary background material. We also present some preliminary results in Section II which are utilized throughout the paper.



\subsection{Related work}
\label{subsec:related_work}

In \cite{dimakis}, Dimakis et al. characterize the information theoretic trade-off between repair bandwidth and per node storage for DSS satisfying the `any $k$ out of $n$' property. The authors map the life span of DSS with a given set of node failures to a multicast problem over a dynamic network. Using this mapping, the authors show that network coding based storage schemes achieve the lower bound on repair bandwidth allowing functional repair \cite{dimakis}. \cite{WuDimRam07} and \cite{Wu2010} present coding schemes that achieve the lower bound on repair bandwidth. The work in \cite{WuDim09, SRKR_itw10, SuhRam_isit10} devise codes that achieve the lower bound derived in \cite{dimakis} when data is downloaded from all surviving nodes in order to perform exact repair of a failed node. The coding schemes in \cite{WuDim09} and \cite{SRKR_itw10, SuhRam_isit10} are tailored for $k < 3$ and $k \leq \frac{n}{2}$, respectively. In \cite{RSK11}, Rashmi et al. design exact-repairable codes, which allow node repair to be performed by contacting $d \leq n-1$ surviving nodes. These codes are optimal for all parameters $(n,k,d)$ at the minimum bandwidth regeneration (MBR) point of the storage-bandwidth trade-off. At the minimum storage regeneration (MSR) point, these codes belong to low rate regime, as their rate is upper bounded by $\frac{1}{2} + \frac{1}{2n}$. 
Recently, researchers have devised high rate exact-repairable codes for the MSR point. \cite{hadamard_dimitris} presents codes for DSS with two parity nodes, which accomplish exact regeneration while being optimal in repair bandwidth. In \cite{cadambe_isit11} and \cite{zigzag13}, permutation-matrix based codes are designed to achieve the bound on repair bandwidth for systematic node repair for all $(n,k)$ pairs. \cite{zigzag_allerton11} further generalizes the idea of \cite{zigzag13} to get MDS array codes for DSS that allow optimal exact regeneration for parity nodes as well.

Towards obtaining coding schemes with small locality, Oggier et al. present coding schemes which facilitate local node repair in \cite{oggier_hom, oggier_proj}. In \cite{Gopalan12}, Gopalan et al. establish an upper bound on the minimum distance of locally repairable linear scalar codes, which is analogous to the Singleton bound. They also show that Pyramid codes proposed in \cite{pyramid} achieve this bound. Subsequently, the work by Prakash et al. extends the bound to a more general definition of locally repairable scalar linear codes \cite{PKLK12} with multiple local parities per a local group. In \cite{DimDim12}, Papailiopoulos et al. generalize the bound in \cite{Gopalan12} to vector codes (potentially non-linear) having one local parity in each local group, and establish per node storage vs. resilience trade-off. They also present locally repairable coding schemes that  exhibit `$k$ out of $n$' property at the cost of small amount of excess storage space per node.

The problem of designing secure DSS against eavesdropping attack has been addressed in \cite{pawar11}. In \cite{pawar11}, Pawar et al. consider an eavesdropper which can get access to the data stored on $\ell$ $(<k)$ storage nodes of DSS operating at the MBR point. The authors derive an upper bound on the secrecy capacity of such systems, and present a coding scheme in the ``bandwidth limited regime'' that achieve this bound. Shah et al. consider the design of secure regenerating codes at the MSR point as well~\cite{SRK_globecom11}. Since the amount of data downloaded for node repair at the MSR point is more than what is eventually stored on the repaired node, the eavesdropper may obtain more information if it is able to access the data downloaded during node repair process. Therefore, at the MSR point, the eavesdropper is modeled as accessing the data stored on $\ell_1$ nodes and data downloaded during $\ell_2$ node repairs (corresponding to distinct nodes), with $\ell_1 + \ell_2 < k$. In~\cite{SRK_globecom11}, Shah et al. present a coding scheme, based on product matrix codes \cite{RSK11}, that achieves the bound on secrecy capacity derived in \cite{pawar11}. They further use product matrix codes based solution for the MSR point as well, which matches the bound in \cite{pawar11} only when $\ell_2 = 0$. Thus, the secrecy capacity for MSR codes is considered to be open when the eavesdropper is allowed to observe downloaded information during node repairs. Moreover, the solution at the MSR point in \cite{SRK_globecom11} gives only low rate schemes as product matrix codes are themselves low rate codes.

There is a closely related line of work on analyzing DSS in the presence of active attacks where an adversary is allowed to modify the content stored on a certain number of nodes throughout the life span of the DSS. Under active attacks, the objective of designing coding schemes to allow for successful decoding of the original data at a data collector, even in the presence of erroneous data injected by the active adversary, is studied in \cite{pawar11, RSRK_isit12, SRV12}.



\section{System Model and Preliminaries}
\label{sec:sys_model}

We consider a DSS with $n$ live nodes at a time. Let $\fv$ denote a file of size $\mathcal{M}$ over a finite field $\mathbb{F}$ that needs to be stored on the DSS. Each node in the system is assumed to store $\alpha$ symbols over $\mathbb{F}$. In what follows, we refer to $\alpha$ as a node size. In order to store the file $\fv = (f_1, f_2,\ldots, f_{\cM})\in \F^{\cM}$ on the DSS, it is first encoded to $n$ blocks, $\xv = (\xv_1, \xv_2,\ldots, \xv_n)\in (\F^{\alpha})^n$, each $\alpha$ symbol long over $\F$. The vector $\xv$ of length $n\alpha$ over $\F$ can be considered as a codeword of a vector code\footnote{Vector codes are formally defined in Section~\ref{subsec:vec}.} $\Cc$. The size of a finite field $\F$ is specified later in the context of various code constructions presented in this paper. Given the codeword $\xv$, node $i$ in an $n$-node DSS stores encoded block $\xv_i$. In this paper, we use $\xv_i$ to represent both block $\xv_i$ and a storage node storing this encoded block.

Given this setup, the network evolves over failures and repairs. In the event of failure of $i$-th storage node, a new node, namely the newcomer, is introduced to the system. This node contacts to $d$ surviving storage nodes and downloads $\beta$ symbols from each of these nodes, which translates to repair bandwidth $\gamma = d\beta$. The newcomer node uses these $d\beta$ number of downloaded symbols to regenerate $\alpha$ symbols $\xv_i$ and stores these symbols. This exact repair process is assumed to preserve `any $k$ out of $n$' property of the DSS, i.e., data stored on any $k$ nodes (potentially including the nodes that are repaired) allows the original file $\fv$ to be reconstructed. In the sequel, we use {\em (n,k)-DSS} to refer to the systems with `any $k$ out of $n$' property.

During the node repair process, the symbols transmitted from node $i$ in order to repair node $j$ are denoted as $\dv_{i,j}$, and the set $\dv_j$ is used to denote all of the symbols downloaded at node $j$. For linear encoding schemes, we use $\Dc_{i,j}$ to denote the subspace spanned by the symbols $\dv_{i,j}$. $\Dc_j$ then represent the subspace downloaded to node $j$, which have a certain dimension in this subspace representation. For a given set of nodes $\Ac$, we use the notation $\xv_\Ac\triangleq \{\sv_i, i\in\Ac\}$. A similar notation is adopted for the downloaded symbols, and the subspace representation. Throughout the text, we usually stick to the notation of having vectors denoted by lower-case bold letters; and, sets and subspaces being denoted with calligraphic fonts. $[n]$ denotes the set $\{1, 2,\ldots, n\}$. We use $a \mid b$ ($a \nmid b$) to represent that  $a$ divides $b$ ($a$ does not divide $b$). Next, we present a formal description of vector codes.


\subsection{Vector codes}
\label{subsec:vec}

An $(n, M, d_{\min}, \alpha)_{q}$ vector code $\cC\subseteq (\F_q^{\alpha})^n$ is a collection of $M$ vectors of length $n\alpha$ over $\F_{q}$. A codeword $\cv \in \cC$ consists of $n$ blocks, each of size $\alpha$ over $\F_q$. We can replace each $\alpha$-long block with an element in $\F_{q^{\alpha}}$ to obtain a vector $\cv = (\cv_1, \cv_2,\ldots, \cv_n) \in \F^n_{q^{\alpha}}$. The minimum distance $d_{\min}$ of $\cC$ is defined as the minimum Hamming distance between any two codewords in $\cC$, when codewords are viewed as vectors in $\F^n_{q^{\alpha}}$, i.e.,
\begin{align}
\label{eq:vec_dmin}
d_{\min} = \min_{\cv, \tilde{\cv} \in \cC: \cv \neq \tilde{\cv}}d_{H}(\cv, \tilde{\cv}).
\end{align}
Here, $d_H(\av, \bv) = \sum_{i = 1}^{n}\mathds{1}_{\{\av_i \neq \bv_i\}}$ denotes the Hamming distance between vectors $\av$ and $\bv$  in $\F^n_{q^{\alpha}}$ with $\mathds{1}_{\{\av_i \neq \bv_i\}}$ representing indicator function, which takes value $1$ if $\av_i \neq \bv_i$ and zero otherwise. An alternative definition for the minimum distance of an $(n, M, d_{\min}, \alpha)_{q}$ vector code is as follows~\cite{DimDim12}:

\begin{definition}\label{def:dmin}
Let $\cv$ be a codeword in $\cC$ selected uniformly at random from $M$ codewords. The minimum distance of $\cC$ is defined as
\begin{equation}
\label{eq:def_dmin}
d_{\min} =  n - \max_{\Ac \subseteq [n]: H(\cv_{\Ac}) < \mathcal{M}}|\Ac|,
\end{equation}
where $\Ac = \{i_1,\ldots, i_{|\Ac|}\} \subseteq [n]$, $\cv_{\Ac} = (\cv_{i_1},\ldots, \cv_{i_{|\Ac|}})$, and $H(\cdot)$ denotes $q$-entropy.
\end{definition}
A vector code is said to be maximum distance separable (MDS) code if $\alpha | \log_{q}M$ and $d_{\min} = n - \frac{\log_{q}M}{\alpha} + 1$. Vector codes are  also known as \emph{array} codes~\cite{BlaRoth, CasBruck}. A linear $(n, M, d_{\min}, \alpha)_{q}$ vector (array) code is a linear subspace of $\F_q^{\alpha n}$ of dimension $\cM = \log_{q}M$. We use $[n, \cM, d_{\min}, \alpha]_{q}$ to denote a linear $(n, M, d_{\min}, \alpha)_{q}$ array code. An $[n, \cM, d_{\min}, \alpha]_{q}$ array code is called {\em MDS array code} if $\alpha | \Mc$ and $d_{\min} = n - \frac{\cM}{\alpha} + 1$. Constructions for MDS array codes can be found e.g. in~\cite{BlaBraBruJai, BlaRoth, CasBruck}.

The encoding process of an $(n, M = q^{\cM}, d_{\min}, \alpha)_{q}$ vector code can be summarized by a function
\begin{align}
\label{eq:encoding_function1}
\mathbb{G}: \mathbb{F}^{\cM}_q \rightarrow \left(\mathbb{F}^{\alpha}_q\right)^n.
\end{align}
In general, $\mathbb{G}$ can be an arbitrary function. However, for an $[n, \cM, d_{\min}, \alpha]_{q}$  (linear) array code, the encoding function is defined by an $\cM \times n\alpha$ generator matrix $G = [\gv^1_1, \ldots, \gv^{\alpha}_{1} | \ldots | \gv^1_{n}, \ldots, \gv^{\alpha}_{n}]$ over $\F_q$. Here, $\gv^1_i, \ldots, \gv^{\alpha}_{i}$ represent the encoding vectors associated with $\cv_i$.

In order to store  a file $\fv$ on a DSS using a vector code $\cC$, $\fv$ is first encoded to a codeword $\cv = (\cv_1, \cv_2,\ldots, \cv_n)\in \cC $. Each symbol of the codeword is then stored on a distinct node. In particular, we have $\xv_i = \cv_i$, where $\xv_i$ denotes the content of $i$th node.


\subsection{Information flow graph}

In their seminal work \cite{dimakis}, Dimakis et al. model the operation of DSS using a multicasting problem over an information flow graph. See Fig.~\ref{fig:info_flow} for illustration of an information flow graph. The information flow graph consists of three types of nodes:
\begin{itemize}
\item Source node ($S$): Source node contains $\mathcal{M}$ symbols of the original file $\fv$. The source node is connected to $n$ nodes.
\item Storage nodes ($\xv_i = (x^{\rm in}_i,x^{\rm out}_i)$): Each storage node in DSS is represented by a pair of nodes in the information flow graph: 1) input node $x^{\rm{in}}_i$ and 2) output node $x^{\rm{out}}_i$. Here, $x^{\rm{in}}_i$ denotes the data downloaded by node $i$, whereas $x^{\rm out}_i$ denotes the $\alpha$ symbols actually stored on node $i$. An edge of capacity $\alpha$ is introduced between $x^{\rm{in}}_i$ and $x^{\rm{out}}_i$ in order to enforce the storage constraint of $\alpha$ symbols per node. For a newcomer node $j$, $x^{\rm{in}}_j$ is connected to output nodes ($x^{\rm{out}}$) of $d$ live nodes with links of capacity $\beta$ symbols each, representing the data downloaded during a node repair.
\item Data collector nodes ($\rm{DC}_i$): A data collector contacting $k$ storage nodes in DSS is modeled by a $\rm{DC}$ node contacting output  node ($x^{\rm{out}}$) of $k$ live nodes in the information flow graph by the edges of capacity $\infty$ each.
\end{itemize}

\begin{figure*}[t]
 \centering
 \includegraphics[width=1\textwidth]{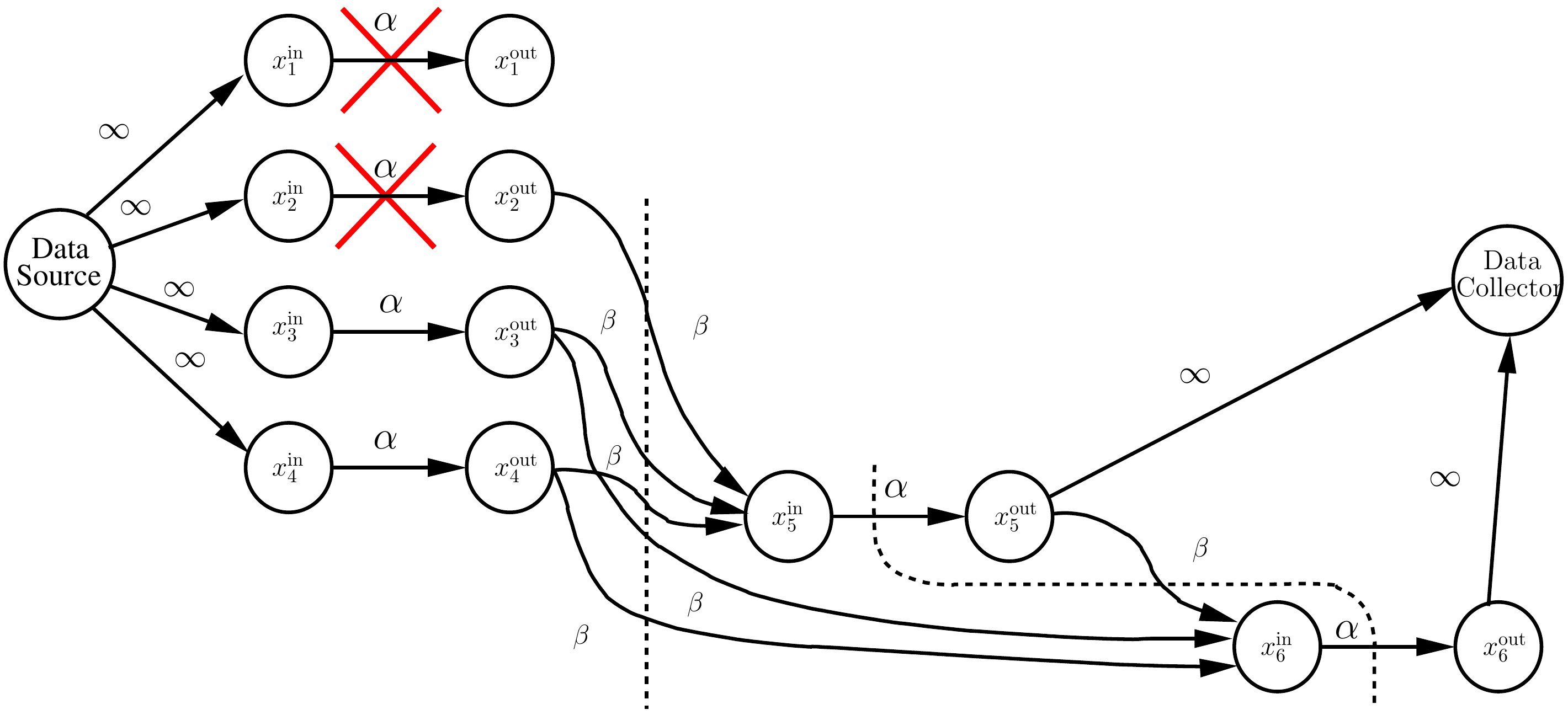}
 \caption{Information flow graph representation of DSS. In this particular example, the system has $4$ live storage nodes at a time while ensuring `any $2$ out of $4$' property. Let node $\xv_1$ fail first. A newcomer $\xv_5$ contacts $\{\xv_2, \xv_3, \xv_4\}$ during the repair of $\xv_1$. In the event of failure of $\xv_2$, nodes in $\{\xv_3, \xv_4, \xv_5\}$ send data to $\xv_6$ for node repair. A data collector contacts $2$ nodes $\{\xv_5, \xv_6\}$ out of $4$ live nodes $\{\xv_3, \xv_4, \xv_5, \xv_6\}$ to reconstruct the information stored on the system. }\label{fig:info_flow}
\end{figure*}

For a given graph $\Gc$ and data collectors $\rm{DC}_i$, the file size stored on such a DSS can be bounded using the max-flow-min cut theorem for multicasting using network coding~\cite{Ahlswede:Network00}.
\begin{lemma}[Max-flow-min-cut theorem for multicasting~\cite{Ahlswede:Network00,dimakis}]
$$\Mc \leq \min_{\Gc} \min_{\rm{DC}_i} \rm{max flow}(S \to \rm{DC}_i,\Gc),$$
where $\rm{flow}(S \to \rm{DC}_i,\Gc)$ represents the flow from the source node $S$ to data collector $\rm{DC}_i$ over the graph $\Gc$.
\end{lemma}
Therefore,  a file of size $\mathcal{M}$  can be delivered to a data collector $\rm{DC}$, only if the min-cut is at least $\mathcal{M}$. In \cite{dimakis},  Dimakis et al. consider $k$ successive node failures and evaluate the min-cut over possible graphs, and obtain the bound given by
\begin{eqnarray}
\mathcal{M} \leq \sum_{i=0}^{k-1}\min\{(d - i)\beta, \alpha\}.\label{eq:dimakis_thm}
\end{eqnarray}
This bound can be achieved by employing linear codes, linear network code in particular. The codes that attain the bound in (\ref{eq:dimakis_thm}) are known as \emph{regenerating codes }\cite{dimakis}. Given a file size $\mathcal{M}$, a trade-off between storage per node $\alpha$ and {\em repair bandwidth} $\gamma \triangleq d\beta$ can be established from \eqref{eq:dimakis_thm}. Two classes of codes that achieve two extreme points of this trade-off are known as {\em minimum storage regenerating} (MSR) codes and {\em minimum bandwidth regenerating} (MBR) codes. The former is obtained by first choosing a minimum storage per node $\alpha=\Mc /k$, and then minimizing $\gamma$ satisfying \eqref{eq:dimakis_thm}; on the other hand, the latter is obtained by first finding the minimum possible $\gamma$, and then finding the minimum $\alpha$ in~(\ref{eq:dimakis_thm}). For MSR codes, we have:
\begin{equation}
\label{eq:msr_point}
(\alpha_{\rm{msr}}, \beta_{\rm{msr}}) = \left(\frac{\mathcal{M}}{k}, \frac{\mathcal{M}}{k(d-k+1)}\right).
\end{equation}
On the other hand, MBR codes are characterized by
\begin{equation}
\label{eq:mbr_point}
(\alpha_{\rm{mbr}}, \beta_{\rm{mbr}}) = \left(\frac{2\mathcal{M}d}{k(2d - k +1)}, \frac{2\mathcal{M}}{k(2d-k+1)}\right).
\end{equation}
For a given DSS with $d\leq n-1$, it can be observed that having $d=n-1$ reduces the repair bandwidth for both the MSR and the MBR codes. Though the bound in \eqref{eq:dimakis_thm} is derived for functional repair, the bound applies to codes with exact repairs as well. The achievability of the bound for all possible and a subset of parameters has been shown at MBR and MSR points, respectively. (See related work in Section~\ref{subsec:related_work} for details.)

Note that regenerating codes are vector codes with additional parameters $d$ (number of nodes contacted during a node repair) and $\beta$ (amount of data downloaded from each contacted node during node repair process). In what follows, we use a modified notation for regenerating codes which highlights these two parameters as well. We refer to regenerating codes as an $(n,\cM, d_{\min},\alpha, \beta, d)$ regenerating code, where $\Mc$ denotes the file size. In particular, an MSR code associated with an $n$ node DSS storing $\Mc$ symbols long file is referred to as an $(n, \cM, d_{\min}=n-\frac{\cM}{\alpha}+1,\alpha, \beta, d)$ MSR code.

\subsection{Eavesdropper model}
\label{subsec:eavesdropper}

In this paper, we consider the eavesdropper model defined in \cite{SRK_globecom11}, which generalizes the eavesdropper model considered in \cite{pawar11}. In \cite{pawar11}, Pawar et al. consider a passive eavesdropper with access to the data stored on $\ell$ $(<k)$ storage nodes. The eavesdropper is assumed to know the coding scheme employed by the DSS. At  the MBR point, a newcomer downloads $\alpha_{\rm{mbr}}=\gamma_{\rm{mbr}}=d\beta_{\rm{mbr}}$ amount of data. Thus, an eavesdropper does not gain any additional information if it is allowed to access the data downloaded during repair. However, repair bandwidth is strictly greater than the per node storage at the MSR point; therefore, an eavesdropper potentially gains more information if it has access to data downloaded during node repair as well. Motivated by this, we consider an $(\ell_1,\ell_2)$ eavesdropper which can access the data stored on any $\ell_1$  nodes with $\Ec_1$ denoting the set of indices of such nodes, and, additionally, the eavesdropper observes the downloaded data for any $\ell_2$  node repairs. We use $\Ec_2$ to denote the set of indices of $\ell_2$ nodes where the eavesdropper accesses the data  downloaded  during node repair. Hence, the eavesdropper has access to $x_i^{\rm{out}},x_j^{\rm{in}},x_j^{\rm{out}}$ for $i\in\Ec_1$ and $j\in\Ec_2$. We summarize the eavesdropper model together with the definition of achievability of a secure file size in the following.
\begin{definition}[Security against an $(\ell_1,\ell_2)$ eavesdropper]
A DSS is said to achieve a secure file size of $\Mc^s$ against an $(\ell_1,\ell_2)$ eavesdropper, if, for any sets $\Ec_1$ and $\Ec_2$ of size $\ell_1$ and $\ell_2$, respectively, $I(\fv^s;\ev)=0$, where $\fv^s$ denotes the secure information of size $\Mc^s$, and $\ev$ represents the eavesdropper's observation vector given by $\ev\triangleq\{x_i^{\rm{out}},x_j^{\rm{in}},x_j^{\rm{out}}: i\in\Ec_1, j\in\Ec_2\}$.
\end{definition}
Note that this definition coincides with the $\{\ell,\ell'\}$ secure distributed storage system in~\cite{SRK_globecom11}, where $\ell=\ell_1+\ell_2$ and $\ell'=\ell_2$.


We remark that, as it will be clear from the following sections, when we consider storing a secure file size $\Mc^s$ in DSS, the remaining $\Mc-\Mc^s$ symbols can be utilized as an additional data which does not have security constraints. Yet, noting the possibility of storing this insecure data, we will refer to this uniformly distributed part as the random data, which is utilized to achieve security. Throughout the text, we use the following lemma to show that the proposed codes satisfy the secrecy constraints.
\begin{lemma}[Secrecy Lemma~\cite{Wyner:The75, SRK_globecom11}]
\label{thm:SecrecyLemma}
Consider a system with information bits $\fv^{s}$,
random bits $\rv$ (independent of $\fv^{s}$),
and an eavesdropper with observations
given by $\ev$. If $H(\ev)\leq H(\rv)$ and
$H(\rv|\fv^{s},\ev)=0$, then the mutual information leakage to eavesdropper is zero, i.e., $I(\fv^{s};\ev)=0$.
\end{lemma}
\begin{proof}
See Appendix~\ref{app:SecrecyLemma}.
\end{proof}


\subsection{Maximum rank distance codes}
\label{subsec:gabidulin}

The coding schemes presented in this paper involve a pre-coding step, where the file $\fv$ is encoded using a Gabidulin code~\cite{Gab85}. Gabidulin codes are an example of optimal rank-metric codes, called maximum rank distance (MRD) codes~\cite{Gab85,Rot91}. Here, we give a brief description of rank-metric codes and highlight various properties of Gabidulin codes that are useful to understand our code constructions.

Let $\F_{q^m}$ be en extension field of $\F_q$. An element $\nu \in \F_{q^m}$ can be represented as the vector $\nuv$$=(\nu_1,\ldots, \nu_m)^T \in \F_q^m$, such that $\nu=\sum_{i=1}^{m}\nu_ib_i$, for a fixed basis $\{b_1,\ldots, b_m\}$ of the extension field $\F_{q^m}$. Using this, a vector $\bf{v}$ $=(v_1, \ldots, v_N) \in \F_{q^m}^N$ can be represented by an $m \times  N$ matrix $\Vm=[v_{i,j}]$ over $\F_q$, which is obtained by replacing each $v_i$ of $\vv$ by its vector representation $(v_{i,1},\ldots,v_{i,m})^T$.

\begin{definition}
The \emph{rank} of a vector $\vv \in \F_{q^m}^N$, denoted by $\rm{rank}(\vv)$ is defined as the rank of its $m\times N$ matrix representation $\Vm$ (over $\F_q$). Similarly, for two vectors $\vv,\uv \in \F_{q^m}^N$, the \emph{rank distance} is defined by
$$d_R(\vv,\uv)=\rm{rank}(\Vm-\Um).$$
\end{definition}
An $[N,K,D]_{q^m}$ \textmd{rank-metric code} $\cC \subseteq \F_{q^m}^N$ is a linear block code over $\F_{q^m}$ of length $N$ with dimension $K$ and minimum rank distance $D$. A rank-metric code that attains the Singleton bound $D \leq N-K+1$ in rank-metric is called \emph{maximum rank distance} (MRD) code. For $m \geq  N$, a family of MRD codes, called Gabidulin codes, was presented by Gabidulin~\cite{Gab85}. Similar to Reed-Solomon codes, Gabidulin codes can be obtained by evaluation of polynomials, however, for Gabidulin codes a special family of polynomials, called \emph{linearized polynomials}, is used:

\begin{definition}
A linearized polynomial $\mathpzc{f}(y)$ over $\F_{q^m}$ of $q$-degree $t$ has the form
$$\mathpzc{f}(y)=\sum_{i=0}^{t}a_iy^{q^i},$$
where $a_i \in \F_{q^m}$, and $a_{t} \neq 0$.
\end{definition}

Encoding a message $(f_1, f_2,\ldots, f_{K})$ to a codeword of an $[N, K, D = N - K +1]$ Gabidulin code over $\mathbb{F}_{q^m}$ ($m \geq N$) consists of two steps~\cite{Gab85}:
\begin{itemize}
\item \textbf{Step~1:} Construct a data polynomial $\mathpzc{f}(y) = \sum_{i = 1}^{K}f_{i}y^{q^{i-1}}$ over $\mathbb{F}_{q^m}$.
\item \textbf{Step~2:} Evaluate $\mathpzc{f}(y)$ at $\{y_1, y_2,\ldots, y_N\}~\subset \mathbb{F}_{q^m}$, $N$ linearly independent (over $\F_{q}$) points from $\F_{q^m}$, to obtain a codeword $\cv= (\mathpzc{f}(y_1),\ldots, \mathpzc{f}(y_N))\in \F_{q^m}^N$.
\end{itemize}
\begin{remark}
\label{rem:linearized_poly}
Note that evaluation of a linearized polynomial is an $\F_{q}$-linear transformation from $\F_{q^m}$ to itself~\cite{MWSl78}. In other words, for any $a, b \in \F_q$ and $\nu_1, \nu_2 \in \F_{q^m}$, we have
$$\mathpzc{f}(a\nu_1 + b\nu_2)= a\mathpzc{f}(\nu_1)+ b\mathpzc{f}(\nu_2).$$
\end{remark}

\begin{remark}
\label{rem:gabidulin_rec}
Given evaluations of $\mathpzc{f}(\cdot)$ at any $K$ linearly independent (over $\mathbb{F}_q$) points in $\mathbb{F}_{q^m}$, say $(z_1,\ldots, z_{K})$, one can get evaluations of $\mathpzc{f}(\cdot)$ at $q^{K}$ points spanned by $\mathbb{F}_q$-linear combinations of $(z_1,\ldots, z_{K})$ using $\F_q$-linearity (a.k.a. linearized property) of $\mathpzc{f}(\cdot)$. (See Remark~\ref{rem:linearized_poly}.) This allows one to recover polynomial $\mathpzc{f}(\cdot)$ of degree $q^{K - 1}$, and therefore, to reconstruct data vector $(f_1,\ldots, f_{K})$ by performing polynomial interpolation. Therefore, an $[N, K, D]$ Gabidulin code (an MRD code in general) is an MDS code and thus can correct any  $D-1=N-K$ erasures (symbols erasures in the codeword, when it is considered as a vector over $\F_{q^m}$ or column erasures when it is considered as a matrix over $\F_q$). We refer to these erasures as \emph{rank erasures} to distinguish from \emph{node erasures} (failures) in the sequel.
\end{remark}

As described in Remark~\ref{rem:gabidulin_rec}, $\F_q$-linearity of the polynomials used in encoding process of Gabidulin codes allows for the recovery of the original message from a sufficient number of evaluations of the polynomial at linearly independent points. Note that these evaluations points may potentially be different from the points used during the encoding process. This makes Gabidulin codes an attractive option to utilize for pre-coding in the coding schemes for DSS presented in this paper, where additional linear operations need to be performed on a codeword from Gabidulin code to meet certain system requirements, e.g., small locality. The importance of this property of Gabidulin codes becomes clear when we present analyses of our different coding schemes in the remaining sections of this paper. Towards this, we present the following result for linearized polynomials that we utilize later to study various code constructions proposed in this paper.

\begin{lemma}
\label{lem:linearized_property}
Let $\pv = (\mathpzc{f}(y_1), \mathpzc{f}(y_2),\ldots, \mathpzc{f}(y_{t\alpha}))\in \F_{q^m}^{t\alpha}$ be the vector containing evaluation of a linearized polynomial $\mathpzc{f}(\cdot)$ over $\F_{q^m}$ at  $t\alpha$ linearly independent (over $\F_{q}$) points $\{y_1, y_2,\ldots, y_{t\alpha}\}$, and let $\hat{\cv} = \pv\hat{G}$ be the codeword of an $[\hat{n}, t\alpha, d_{\min} = \hat{n} - t + 1, \alpha]_q$ MDS array code over $\F_q$ (with $\hat{G}\in \F_q^{t\alpha\times\hat{n}\alpha }$ as its generator matrix) corresponding to message vector $\pv$. Then, any $s$ (out of $\hat{n}$) vector symbols of $\hat{\cv} = (\hat{\cv}_1,\ldots, \hat{\cv}_{\hat{n}})\in (\F_{q^m}^{\alpha})^{\hat{n}}$ corresponds to evaluations of $\mathpzc{f}(\cdot)$ at $\min\{s, t\}\alpha$ linearly independent (over $\F_q$) points  from the subspace spanned by $\{y_1, y_2,\ldots, y_{t\alpha}\} \subset \F_{q^m}$. 
\end{lemma}

\begin{proof}
See Appendix~\ref{app:linearized_property}.
\end{proof}

The efficient decoding algorithms for Gabidulin codes can be found e.g. in~\cite{GaPi08}.

\subsection{Locally repairable codes}
\label{subsec:lrc}

In this subsection, we present a formal definition of locally repairable codes. First, we define the notion of a punctured vector code. Given an $(n, M, d_{\min}, \alpha)_{q}$ vector code $\cC$ ($[n, \Mc = \log_q{M}, d_{\min}, \alpha]_q$ in case of linear code) and a set $\Sc \subset [n]$, we use $\cC|_{\Sc}$ to denote the code obtained by puncturing $\cC$ on $[n]\backslash\Sc$. In other words, codewords of $\cC|_{\Sc}$ are obtained by keeping those vector symbols in $\cv  = (\cv_1,\ldots, \cv_{n}) \in \cC$ which have their indices in set $\Sc$.

Next, we present the definition of vector LRCs, which generalizes the definition of \emph{scalar} LRCs ($\alpha = 1$) presented in~\cite{PKLK12}. (See Fig.~\ref{fig:lrc_def}.)

\begin{definition}
\label{def:vectorLRC}
An $(n, M, d_{\min},\alpha)_{q}$ vector code $\cC$ is said to have $(r, \delta)$ \emph{locality} if for each vector symbol $\cv_i \in \F_{q}^{\alpha}$, $i \in [n]$, of a codeword $\cv = (\cv_1,\ldots, \cv_n)\in \cC$,
there exists a set of indices $\Gamma(i)$ such that
\begin{itemize}
\item $i\in\Gamma(i)$
\item $|\Gamma(i)|\leq r+\delta - 1$
\item Minimum distance of $\cC|_{\Gamma(i)}$ is at least $\delta$.
\end{itemize}
An $(n, M, d_{\min}, \alpha)_q$ vector code $\cC$ with $(r, \delta)$ locality is referred to as an $(r, \delta, \alpha)$ locally repairable code (LRC).
\end{definition}

\begin{figure*}[t]
 \centering
 \includegraphics[width=0.9\textwidth]{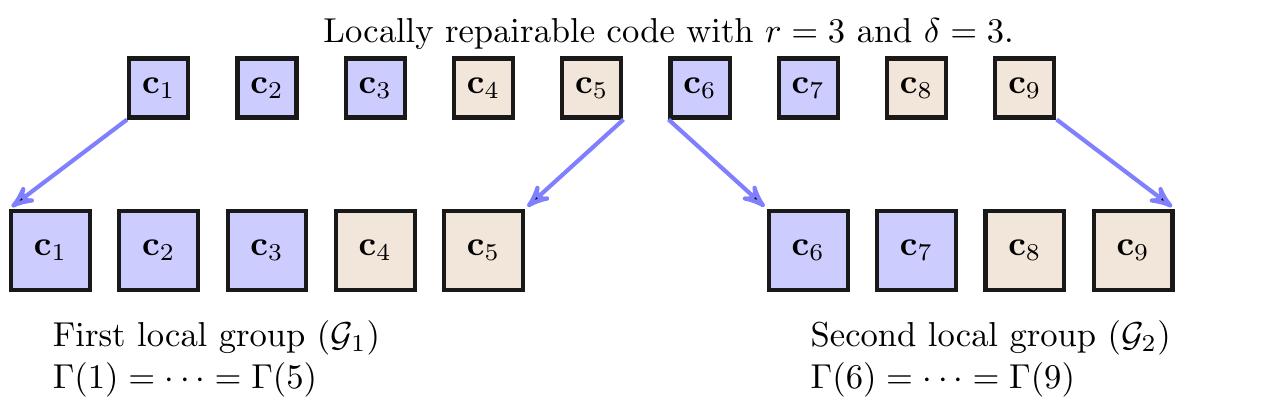}
 \caption{Illustration of a $9$-length $(r = 3, \delta = 3, \alpha)$ locally repairable code. We have $\cv_i \in \F^{\alpha}$, for each $i \in [9]$. The code has two disjoint local groups: $\{\cv_1,\ldots, \cv_5\}$ and $\{\cv_6,\ldots, \cv_9\}$. Brown nodes inside each local group highlight redundancy within the local group.}\label{fig:lrc_def}
\end{figure*}

\begin{remark}
$(r, \delta, \alpha = 1)$ LRCs are named as $(r,\delta)$ scalar LRCs.
\end{remark}

\begin{remark}
\label{rem:lrc_def}
The last requirement in Definition~\ref{def:vectorLRC} implies that each element $j \in \Gamma(i)$ can be written as a function of any set of $r$ elements in $\Gamma(i)\backslash \{j\}$. Moreover, the last two requirements in the definition ensure that $H(\Gamma(i))\leq r\alpha$.
\end{remark}

Distinct sets in $\{\Gamma(i)\}_{i \in [n]}$ are called local groups. Note that two nodes $i_1 \neq i_2$ may have $\Gamma(i_1) = \Gamma(i_2)$. We use $g$ to denote the number of distinct local groups in $\cC$.  (See Fig.~\ref{fig:lrc_def}.) The definition of LRCs presented in this paper generalizes that given in~\cite{DimDim12}, which is restricted to $\delta=2$. In order to store  a file $\fv$ on a DSS using an LRC, $\fv$ is first encoded to a codeword of the LRC. Each symbol of the codeword is then stored on a distinct node. In particular, we have $\xv_i = \cv_i$, where $\xv_i$ denotes the content of $i$th node. It follows from Remark~\ref{rem:lrc_def} that, for an $(r, \delta, \alpha)$ LRC, $d = r$ is possible by using $r$ nodes in $\Gamma(i)\backslash\{i\}$ to repair node $i$.

Note that Definition~\ref{def:vectorLRC} requires that each symbol in codeword $\cv$ is part of a local group, which ensures that every encoded symbol can be repaired in a local manner. This property of a code is termed as {\em all-symbol locality} in the literature. When locality is ensured only for a subset of symbols in $\cv$ which is sufficient to recover the original data, e.g., systematic symbols in a systematic code, the code is said to have {\em information locality}. In this paper, we focus on codes with all-symbol locality.

In \cite{PKLK12}, Prakash et al. present the following upper bound on the minimum distance of an $(r,\delta)$ scalar LRC:
\begin{equation}
\label{scalarUpperBound}
d_{\min}\leq n-\cM+1-\left(\left\lceil\frac{\cM}{r}\right\rceil-1\right)(\delta-1).
\end{equation}

It was established in~\cite{Gopalan12} that Pyramid codes presented in~\cite{pyramid} attain this bound. Note that Pyramid codes have information locality. An explicit construction of optimal scalar LRCs with all-symbols locality is known only for the case $n=\left\lceil\frac{\Mc}{r}\right\rceil(r+\delta-1)$~\cite{PKLK12, HLM2007}. Towards optimal scalar LRCs for broader range of parameters, \cite{PKLK12} establishes the existence of scalar codes with all-symbols locality for the setting when $(r+\delta-1)|n$, assuming the field size $|\F|>\cM n^{\cM}$. In this paper, we provide an explicit construction of optimal scalar LRCs  with all-symbols locality for even broader range of system parameters by relaxing the restriction of $(r+\delta-1)|n$.

Papailiopoulos et al.~\cite{DimDim12} present the following upper bound on the minimum distance of an $(r,\delta = 2,\alpha)$ LRC.
\begin{equation}
\label{eq:dimitrisBound}
d_{\min}\leq n-\left\lceil\frac{\cM}{\alpha}\right\rceil-\left\lceil\frac{\cM}{r\alpha}\right\rceil+2
\end{equation}
In \cite{DimDim12}, the authors also propose a construction of LRCs with $\delta = 2$ that achieve the bound in \eqref{eq:dimitrisBound}. In this paper, we generalize this bound for any $\delta \geq 2$ and present $(r, \delta,\alpha)$ LRCs with multiple local parity nodes that are optimal with respect to the generalized bound.

Here, we highlight a couple of advantages of LRCs with multiple parity nodes per local group. These codes may allow system administrators to deal with variability in response times in a large scale storage system. In a large scale system, storage nodes exhibit variable response time due to multiple reasons, including resource sharing, software updates, maintenance activities, and energy maintenance (turning on a node in power saving mode introduces latency in its response time)~\cite{tail_scale}. In particular, some of the contacted nodes can take longer than expected time to respond with data during node repair (or direct data access from systematic nodes). In order to mitigate this effect, an administrator can send data request to more than required number of nodes in a local group for node repair, and complete the repair process as soon as the required number of nodes from the set of contacted nodes have responded with data. Similar ideas to reduce response time for file reconstruction with MDS codes have been proposed in~\cite{tail_scale, GLS13}. Besides this, LRCs having multiple local parities per local group exhibit a stronger resilience to eavesdropping attacks. In particular, as detailed in Section~\ref{sec:secrecy_local}, both scalar LRCs and vector LRCs with single local parity have poor secrecy rate in the presence of a passive eavesdropper.


\section{Secrecy in repair bandwidth efficient DSS}
\label{sec:secrecy_BW}

Considering that the eavesdropped nodes may not carry secure information to the data collectors in the bound given by $\eqref{eq:dimakis_thm}$, \cite{pawar11} establishes the following upper bound on the secure file size when the eavesdropper observes the content of $\ell$ nodes.
\begin{equation}
\label{eq:pawar_bound}
\mathcal{M}^{s} \leq \sum_{i = \ell+1}^{k}\min\{(d-i+1)\beta,\alpha\}.
\end{equation}
Pawar et al. show that this bound is tight in the {\em bandwidth limited regime}, $\gamma \leq \Gamma = (n-1)\alpha$ with $d = n-1$, by presenting a coding scheme that is secure against the passive eavesdropper observing $\ell$ storage nodes. This scheme essentially corresponds to the MBR point (see (\ref{eq:mbr_point})) where a newcomer node contacts all the remaining nodes. \cite{SRK_globecom11} proposes product matrix based secure coding schemes achieving this bound for any $\ell$ at the MBR point. However, the coding scheme proposed in \cite{SRK_globecom11} can only store a secure file size of $(k-\ell_1-\ell_2)(\alpha - \ell_2\beta)$ at the MSR point. At the MSR point, the bound in (\ref{eq:pawar_bound}) reduces to
$$\Mc^{s}\leq(k-\ell_1-\ell_2)\alpha.$$
As a result, \cite{SRK_globecom11} concludes that the proposed scheme based on product matrix codes achieves secrecy capacity only when $\ell_2=0$. This corresponds to the scenario for which the eavesdroppers are not allowed to observe downloaded packets. This leaves the following questions open:
\begin{itemize}
\item Can bound (\ref{eq:pawar_bound}) be further tightened for MSR point?
\item Is it possible to obtain secure coding schemes at the MSR point that have better rate and/or secrecy capacity than that of the scheme proposed in \cite{SRK_globecom11}?
\end{itemize}
In this section, we answer both questions affirmatively. We first derive a generic upper bound on the amount of data that can be securely stored on DSS for bandwidth efficient repairable codes at the MSR point, which also applies to bandwidth efficient exact-repairable code. Next, we prove a result specific to exact-repairable code for $d=n-1$ which allows us to provide an upper bound on the secrecy capacity against an $(\ell_1,\ell_2)$-eavesdropper at the MSR point. This bound is tighter than a bound that can be obtained from the generic bound we provide. We subsequently combine the classical secret sharing scheme due to \cite{shamir79} with an existing class of exact-repairable MSR codes to securely store data in the presence of an $(\ell_1,\ell_2)$-eavesdropper. We show that this approach gives a higher rate coding scheme compared to that of \cite{SRK_globecom11} and achieves the secrecy capacity for any $\ell_1$  with $\Ec_2 \subset [k]$ and $\ell_2\leq 2$.


\subsection{Improved bound on secrecy capacity at the MSR point}
\label{sec:BW_improved}

In order to get the desired bound, we rely on the standard approach of computing a cut in the information flow graph associated with DSS. We consider a particular pattern of eavesdropped nodes, where eavesdropper observes content put on $\ell_1$ nodes and data downloaded during first $\ell_2$ node failures that do not involve already eavesdropped $\ell_1$ nodes. Using the max-flow min-cut theorem, this case translates into an upper bound on the secrecy capacity for an $(n,k)$-DSS operating at the MSR point (see (\ref{eq:msr_point})).


\begin{figure*}[t!]
 \centering
 \includegraphics[width= 0.8\textwidth]{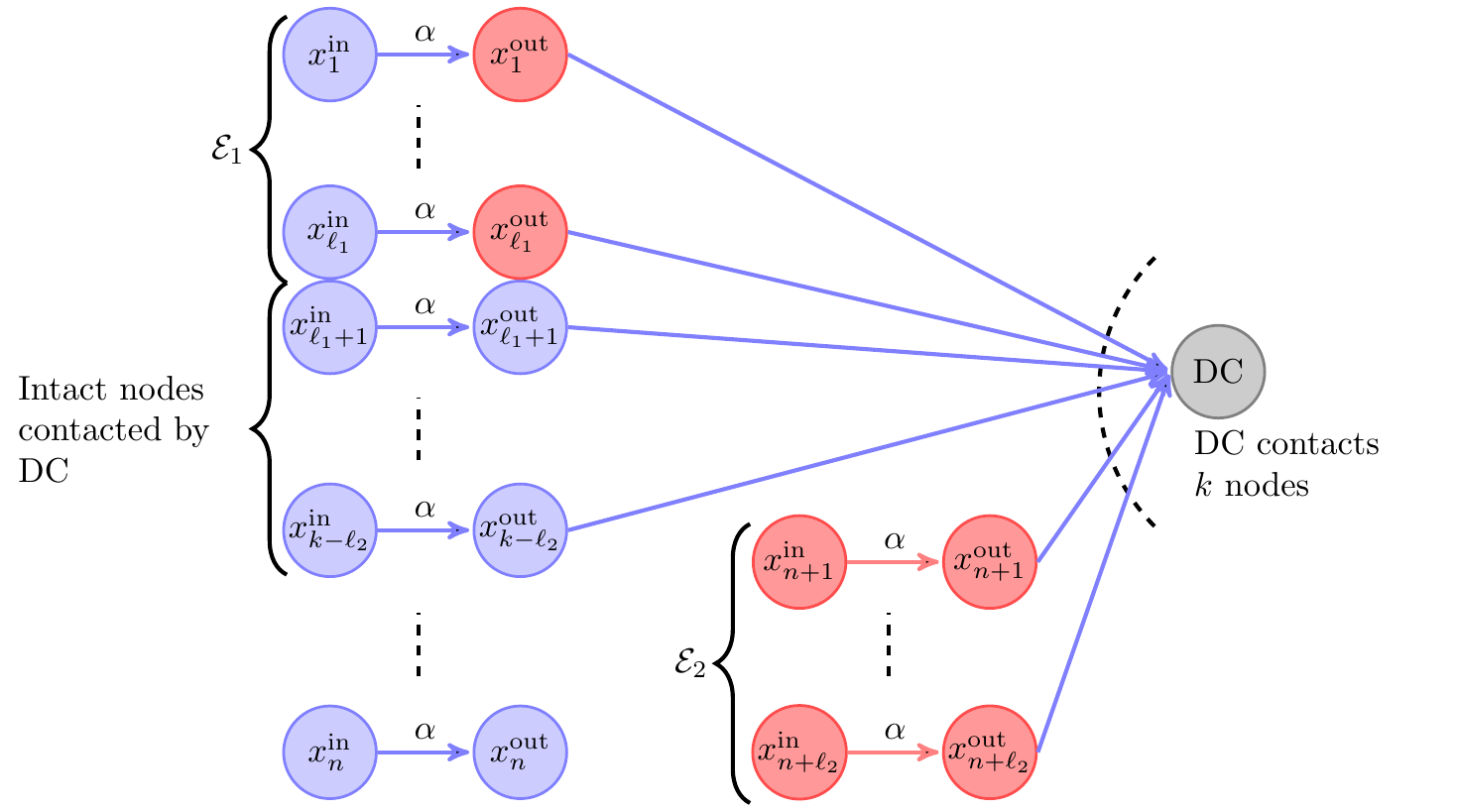}
 \caption{An information flow graph associated with an $(n,k)$-DSS in the presence of an $(\ell_1, \ell_2)$-eavesdropper. Red colored nodes are observed by the eavesdropper. Note that we have $\xv_i = (x^{\rm in}_i, x^{\rm out}_i)$ in an information flow graph representation for a DSS. In particular, we have $\Ec_1 = \{\xv_1,\ldots, \xv_{\ell_1}\}$ and $\Ec_2 = \{\xv_{n+1}, \ldots, \xv_{n+\ell_2}\}$. Here, we assume that $\xv_{k-\ell_2+1}, \ldots, \xv_{k}$ fail subsequently in the order specified by their indices and are repaired by introducing nodes $\xv_{n+1},\ldots, \xv_{n+\ell_2}$ respectively. Data collector (DC) contacts $\xv_1,\ldots, \xv_{k-\ell_2}, \xv_{n+1},\ldots, \xv_{n+\ell_2}$ to reconstruct the original data stored on the DSS.}\label{fig:upperbound1}%
\end{figure*}

\begin{theorem}
\label{thm:sec_cap}
For an $(n,k)$-DSS employing an MSR code, we have
\begin{align}
\label{eq:sec_cap_gen0}
\Mc^{s} \leq \sum^{k-\ell_2}_{i = \ell_1+1} \left(\alpha - I(\xv_{i}; \dv_{i,n+1},\ldots, \dv_{i,n+\ell_2}) \right),
\end{align}
where $\xv_{i}$ and $\dv_{i,j}$ denote the data stored on node $i$ and the data downloaded from node $i$ to perform node repair at node $j$, respectively.
For a linear coding scheme, the bound in \eqref{eq:sec_cap_gen0} reduces to the following.
\begin{align}
\label{eq:sec_cap_gen}
\Mc^{s} \leq \sum^{k-\ell_2}_{i=\ell_1+1}\left(\alpha - \dim \left(\sum_{j=1}^{\ell_2} \Dc_{i,n+j}\right)\right).
\end{align}
Here, $\Dc_{i, j}$ denote the subspace spanned by the encoding vectors associated with symbols transferred from node $i$ to node $j$ during node repair.
\end{theorem}

\begin{proof}
Consider Fig.~\ref{fig:upperbound1}, which describes a particular case that may arise during the lifespan of a DSS. Here, $\xv_1, \xv_2,\ldots, \xv_n$ represent the original $n$ storage nodes in DSS as defined in Sec.~\ref{sec:sys_model}. Assume that nodes $\xv_{k-\ell_2+1}, \ldots, \xv_{k}$ subsequently fail in the order specified by their indices. These $\ell_2$ failures are repaired by introducing nodes $\xv_{n+1},\ldots, \xv_{n+\ell_2}$ in the system following a node repair process associated with the coding scheme employed by the DSS. Consider $\Ec_1 = \{\xv_1,\ldots, \xv_{\ell_1}\}$ as the set of $\ell_1$ nodes where eavesdropper observes the stored content, and let $\Ec_2 = \{\xv_{n+1},\ldots, \xv_{n+\ell_2}\}$ be the set of nodes which are exposed to the eavesdropper during their node repair, allowing eavesdropper to have access to all the data downloaded during node repair of set $\Ec_2$.
Let $\Rc$ denote the set of  $k -(\ell_1+\ell_2)$ remaining original nodes $\{\xv_{\ell_1+1},\ldots, \xv_{k-\ell_2}\}$, which are not observed by the eavesdropper directly, and information stored on these nodes may leak to eavesdroppers only when these nodes participate in node repair.
Assume that a data collector contacts a set of $k$ nodes given by $\Kc = \Ec_1 \cup \Ec_2 \cup \Rc$ in order to reconstruct the original data.  For a file $\fv^s$ to be securely stored on the DSS, we have
{\allowdisplaybreaks
\begin{align}
\Mc^{s} = H(\fv^s) &= H(\fv^s|\xv_{\Ec_1}, \dv_{\Ec_2}) \label{eq:sec} \\
&= H(\fv^s|\xv_{\Ec_1}, \dv_{\Ec_2}) - H(\fv^s|\xv_{\Ec_1}, \xv_{\Ec_2}, \xv_{\Rc}) \label{eq:rec} \\
&\leq  H(\fv^s|\xv_{\Ec_1}, \dv_{\Ec_2}) - H(\fv^s|\xv_{\Ec_1}, \dv_{\Ec_2}, \xv_{\Rc}) \n \\
& =  I(\fv^s; \xv_{\Rc}|\xv_{\Ec_{1}}, \dv_{\Ec_{2}}) \n \\
&\leq H(\xv_{\Rc}|\xv_{\Ec_1}, \dv_{\Ec_2}) \n \\
& \leq  H(\xv_{\Rc}|\dv_{\Ec_2}) \n \\
&= \sum^{k-\ell_2}_{i = \ell_1+1} H(\xv_{i}|\xv_{\ell_1+1},\ldots, \xv_{i-1}, \dv_{\Ec_2}) \n \\
&\leq \sum^{k-\ell_2}_{i = \ell_1+1} H(\xv_{i}|\dv_{i,n+1},\ldots, \dv_{i,n+\ell_2}) \n \\
& = \sum^{k-\ell_2}_{i = \ell_1+1} \left(\alpha - I(\xv_{i}; \dv_{i,n+1},\ldots, \dv_{i,n+\ell_2}) \right).
\end{align}}

Here, (\ref{eq:sec}) follows from the fact that coding scheme employed in DSS is secure against an $(\ell_1, \ell_2)$-eavesdropper, i.e., $I(\fv^s;\xv_{\Ec_1},\dv_{\Ec_2}) = H(\fv^s) - H(\fv^s|\xv_{\Ec_1},\dv_{\Ec_2}) = 0$, and (\ref{eq:rec}) is a consequence of the fact that the original file can be recovered from data stored on any set of $k$ nodes.

For linear coding schemes, $I(\xv_{i}: \dv_{i,n+1},\ldots, \dv_{i,n+\ell_2})$ can be replaced by $\dim \left(\sum_{j=1}^{\ell_2}\Dc_{i,n+j}\right)$ to obtain
\begin{align}
H(\fv^{s}) \leq \sum^{k-\ell_2}_{i=\ell_1+1}\left(\alpha - \dim \left(\sum_{j=1}^{\ell_2}\Dc_{i,n+j}\right)\right)
\end{align}
\end{proof}

In Theorem~\ref{thm:sec_cap}, $\dim \left(\sum_{j=1}^{\ell_2}D_{i,n+j}\right)$ is lower bounded by $\beta$, and using this we obtain the following corollary.

\begin{corollary}\label{thm:alpha-beta_bound}
For an $(n,k)$-DSS employing an $(n,\cM, d_{\min}=n-\frac{\cM}{\alpha}+1,\alpha, \beta, d)$ MSR code, we have
\begin{equation}
\label{eq:secrecy_triv_bound}
\Mc^{s} \leq (k - \ell_1 - \ell_2)(\alpha - \beta).
\end{equation}
\end{corollary}
This shows that the secure code construction proposed in \cite{SRK_globecom11} is optimal for $\ell_2 = 1$.

The following lemma is specific to exact-repairable linear MSR codes that employ interference alignment for node repair with $d=n-1$. It is shown in \cite{SRKR12_IA} that interference alignment is a necessary component of an exact-repairable linear scalar ($\beta = 1$) MSR code. The necessity of interference alignment holds for $\beta > 1$ as well. Therefore, the following bound is fairly general and applies to all known exact-repairable  codes at the MSR point. Following the standard terminology in DSS literature, node $i$ has $\beta \times \alpha$ repair matrices $\{V_{i,j}\}$ associated with remaining nodes $j \neq i$. In the event of failure of node $j$, a newcomer downloads $V_{i,j}\xv_i^T$ from every node $i$, $i\neq j$. We use $V_{i,j}$ to denote both matrix $V_{i,j}$ and row-space of $V_{i,j}$.
\begin{lemma}
\label{lem:intersection}
Consider an $(n,k)$-DSS employing a systematic MSR code with $(n-k)$ linear parity nodes. Assume that $d = n-1$, i.e., all the remaining nodes are contacted to repair a failed node. Let $V_{i,j}$ be the repair matrices associated with node $i$, which is used to perform node repair for node $j$. Then, for each $i \in [k]$, i.e., for systematic nodes, we have
\begin{equation}
\dim\left(\bigcap_{j \in \mathcal{A}}V_{i,j}\right) = {\rm{rank}}\left(\bigcap_{j \in \mathcal{A}}V_{i,j}\right) \leq \frac{\alpha}{(n-k)^{|\mathcal{A}|}},
\end{equation}
where $\mathcal{A} \subseteq [k]\backslash\{i\}$.
\end{lemma}

\begin{proof}
See Appendix~\ref{sec:lemma_intersection_proof}.
\end{proof}

\vspace{.5cm}
According to the notation specified in Section~\ref{sec:sys_model}, we have $\dv_{i,j} = V_{i,j}\xv^T_i$. Therefore, subspace $\Dc_{i,j}$ corresponds to the row space spanned by $V_{i,j}[\gv^1_i,\ldots, \gv^{\alpha}_i]^T$. Here, $\gv^1_i,\ldots, \gv^{\alpha}_i$ represent the encoding vectors associated with $\xv_i$, i.e., $\xv_i = \fv[\gv^1_i,\ldots, \gv^{\alpha}_i]$. It follows from the well-known dimension formula for vector spaces that
\begin{align}
\dim\left(\Dc_{i,n+1}+\Dc_{i,n+2}\right) &= \dim\left(\Dc_{i,n+1}\right) + \dim\left(\Dc_{i,n+2}\right) \n \\
&~~~- \dim\left(\Dc_{i,n+1} \cap \Dc_{i,n+2}\right) \n \\
&= \beta + \beta - \dim\left(\Dc_{i,n+1} \cap \Dc_{i,n+2}\right) \n \\
& \geq 2\beta - \frac{\alpha}{(n-k)^{2}} \n \\
& = 2\beta - \frac{\beta}{(n-k)}, \label{eq:intersec_dss}
\end{align}
where (\ref{eq:intersec_dss}) follows from Lemma~\ref{lem:intersection} and the fact that $\beta = \frac{\alpha}{n-k}$ at the MSR point with $d=n-1$. Now combining (\ref{eq:intersec_dss}) with Theorem~\ref{thm:sec_cap}, we get the following corollary:

\begin{corollary}
\label{cor:intersec}
Given an $(n, \alpha k, n-k+1,\alpha, \beta, d=n-1)$ MSR code that employs interference alignment to perform node repair, for $\ell_2 \leq 2$ we have
\begin{eqnarray}
\label{eq:sec_cap_gen_cor}
\Mc^{s} \leq (k-\ell_1-\ell_2)\left(\alpha - \theta\left(\alpha,\beta,\ell_2\right)\right)
\end{eqnarray}
where
\begin{align}
\theta\left(\alpha,\beta,\ell_2\right) = \left\{\begin{array}{cc}
                \beta, & \textmd{ if } \ell_2=1 \\
                2\beta-\frac{\beta}{(n-k)}, & \textmd{ if } \ell_2=2
              \end{array}\right.
\end{align}
\end{corollary}


\subsection{Construction of secure MSR codes with $d=n-1$}
\label{subsec:secure_msr}

In this subsection, we present a construction for an $(\ell_1, \ell_2)$-secure coding scheme when $\Ec_2 \subset \Sc$, for a given set $\Sc$ of size $|\Sc|=k$, i.e., the nodes where an eavesdropper observes all the data downloaded during node repair are restricted to a specific set of $k$ nodes. The construction is based on concatenation of  Gabidulin codes~\cite{Gab85} and zigzag codes~\cite{zigzag13}. The inner code, zigzag code, allows bandwidth efficient repair\footnote{Zigzag codes achieve optimal {\em rebuilding ratio}, i.e., the ratio of $d\beta$ (the amount of data that is read and transferred from the contacted nodes to a newcomer node) to $(n-1)\alpha$ (total data stored on the surviving nodes). Since zigzag codes have optimal rebuilding ratio, they are also repair bandwidth efficient.} of systematic nodes with $d = n-1$.

First, we present a brief description of the construction of  an $(n, k)$ zigzag code~\cite{zigzag13}. Let $p = n - k$ denote the number of parity nodes and $\mathcal{M} = k(n-k)^k = kp^k$ be the size of file that needs to be encoded.
The $kp^k$ message symbols are arranged in a $p^k \times k$ array, and $p^k$ symbols in each of $k$ columns of the array are indexed using integers in $[0, p^{k}-1]$. Note that an integer in $[0, p^{k} - 1]$ can also be associated with a vector in $\mathbb{Z}_{p^k}$. Node $j$ ($j \in [k]$) stores  $p^k$ symbols from $j$-th column of the array.  Let $x_{ij}$ for $j \in [k]$ denote the symbol in node $j$ with index $i$. In order to generate $(l+1)$-th parity node, define zigzag sets $Z^{l}_{t} = \{x_{ij}: i + le_j = t\}$, where $t$ and $l$ ranges over integers in $[0, p^{k} - 1]$ and $[0, p-1]$, respectively. Here, $e_j$ denotes the $j$-th standard basis vector in $\mathbb{Z}_{p^k}$ that has all zero entries except $1$ at $j$-th position. Also, $i$ is used to represent both an integer and its vector representation in $\mathbb{Z}_{p^k}$. Now, $t$-th symbol of $(l+1)$-th parity node is obtained by taking a linear combination (over $\F_q$) of symbols in the zigzag set $Z^{l}_{t}$. Note that each symbol for a systematic node is used in exactly one linear combination stored on each parity node. (The construction described here corresponds to Construction $6$ in \cite{zigzag13} with $m = k$ and $\{v_0, v_1,\ldots, v_{k-1}\} = \{e_1, e_2,\ldots, e_k\}$.) 

The repair of $j$-th systematic node (column) is performed by accessing symbols in $d=n-1$ surviving nodes, where the downloaded symbol index (of each live node) belongs to the set $Y_j = \{i : i\cdot e_j =0\}$. 
Here, $a\cdot b  = \sum_{i}a_ib_i$~(${\rm mod}~p$) denotes the dot product between $a$ and $b$ in $\mathbb{Z}_{p^k}$. We first state the following property associated with the repair process of a zigzag code.
\begin{lemma}
\label{lm:union_zigzag}
Assume that an eavesdropper gains access to the data stored on $\ell_1$ nodes and the data downloaded during node repair of $\ell_2$ systematic nodes in an $(n = k+p, k)$ zigzag code. Then the eavesdropper can only observe
$$kp^k-p^k(k-\ell_1-\ell_2)\left(1-\frac{1}{p}\right)^{\ell_2}$$
independent symbols.
\end{lemma}
\begin{proof}
See Appendix~\ref{appen:union_zigzag}
\end{proof}
\vspace{.3cm}

\begin{figure*}[t]
 \centering
 \includegraphics[width= 0.9\textwidth]{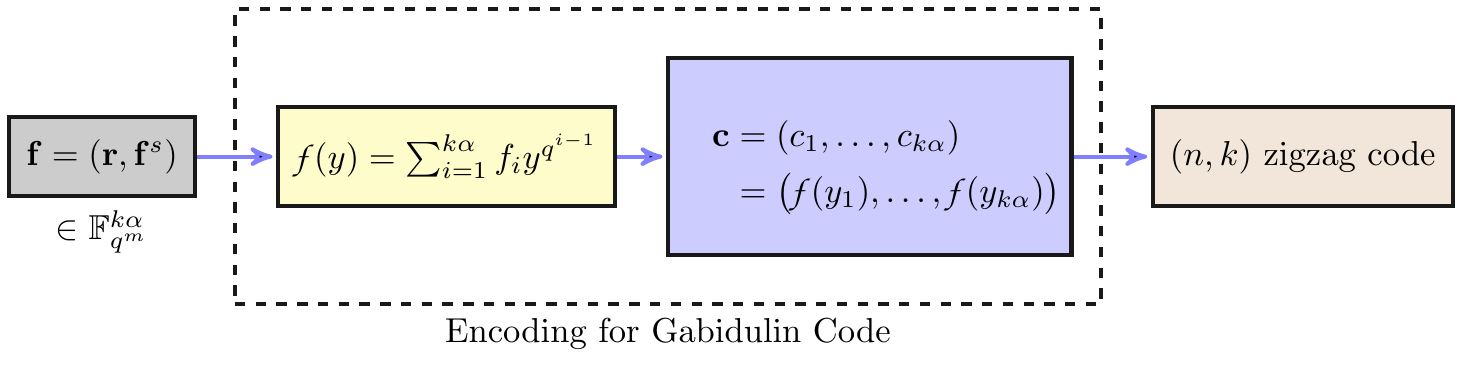}
 \caption{The proposed secure coding scheme at the MSR point.}\label{fig:msrachievability}%
\end{figure*}

We now detail an achievability scheme of this section. (See Fig.~\ref{fig:msrachievability}.) Let $\fv^{s}$ denote secure information of size $p^k(k-\ell_1-\ell_2)(1-\frac{1}{p})^{\ell_2}$ symbols in $\F_q^{m}$. We take $kp^k-p^k(k-\ell_1-\ell_2)(1-\frac{1}{p})^{\ell_2}$ i.i.d. random symbols $\rv = (r_1,\ldots, r_{kp^k-p^k(k-\ell_1-\ell_2)(1-\frac{1}{p})^{\ell_2}})$, distributed uniformly at random over $\F_{q^m}$, and append $\rv$ with $\fv^{s}$ to obtain $\fv = (\rv, \fv^s)\in \F_{q^m}^{k\alpha}$. We then encode $\fv$ in two stages as follows:
\begin{enumerate}
\item First, encode $\fv = (f_1,\ldots, f_{k\alpha})$, $\alpha=p^k$, using the encoding process of a $[k\alpha, k\alpha, 1]_{q^m}$ Gabidulin code as specified in Section~\ref{subsec:gabidulin}. Let
$$
\mathpzc{f}(y)=\sum_{i=1}^{k\alpha} f_i y^{q^{i-1}}
$$
be the underlying linearized polynomial used in the encoding process.
\item Encode the codeword of the Gabidulin code, obtained from the previous step, using a $(k+p, k)$ zigzag code defined over $\F_q$, i.e., the encoding coefficient of the zigzag code are from $\F_q$. Store the output of this step at $n = k+p$ nodes in the storage system.
\end{enumerate}

As described in Remark~\ref{rem:linearized_poly}, $\mathpzc{f}(\cdot)$ is a $\F_q$-linear function; therefore, all the symbols stored on parity nodes of the zigzag code are evaluations of $\mathpzc{f}(\cdot)$ at points in $\F_{q^m}$. (See also Lemma~\ref{lem:linearized_property}).

Next, we present the following result on security of the coding scheme described above, which performs secrecy pre-coding  using a Gabidulin code.
\begin{theorem}
\label{thm:msr_sec_cap}
The code obtained by performing secrecy pre-coding of a zigzag code, as described above, with $\alpha=p^k$ and $d = n - 1$ achieves a secure file size $$\Mc^s=(k-\ell_1-\ell_2)p^k\left(1-\frac{1}{p}\right)^{\ell_2}$$ against an $(\ell_1, \ell_2)$-eavesdropper with $\Ec_2 \subset [k]$. Here, $p=n-k$, and $[k]$ denotes the set of systematic nodes of zigzag code.
In addition, for any $(\ell_1,\ell_2)$ such that $\ell_2 \leq 2$ and $\Ec_2 \subset [k]$, this code attains the upper bound on the secure file size given in Corollary~\ref{cor:intersec}; hence, characterizes the secrecy capacity at the MSR point with $d=n-1$ for a class of $(\ell_1, \ell_2)$-eavesdroppers in which download eavesdropping is restricted to the systematic nodes.
\end{theorem}

\begin{proof}
The repair and data reconstruction properties of the proposed code follow from the construction of zigzag codes~\cite{zigzag13}. We utilize Lemma~\ref{thm:SecrecyLemma} to prove the security of this code against an $(\ell_1, \ell_2)$-eavesdropper with $\Ec_2 \subset [k]$. Considering that $\ev$ denotes the symbols stored on DSS that an eavesdropper observes, we need to show: (i) $H(\ev) \leq H(\rv)$ and (ii) $H(\rv|\ev, \fv^{s})=0$ to establish the security of the proposed coding scheme. It follows from  Lemma~\ref{lm:union_zigzag} that an eavesdropper observes $kp^k-p^k(k-\ell_1-\ell_2)\left(1-\frac{1}{p}\right)^{\ell_2}$ independent symbols, which correspond to evaluations of $\mathpzc{f}(\cdot)$  at linearly independent (over $\F_q$) points in $\F_{q^n}$. Since $|\ev| = kp^k-p^k(k-\ell_1-\ell_2)\left(1-\frac{1}{p}\right)^{\ell_2}$, we have $H(\ev) = H(\rv)$, which is the first requirement to establish the security claim. It remains to show $H(\ev|\rv, \fv^{s})=0$, for which we describe how an eavesdropper can decode random symbols $\rv$ given $\ev$ and $\fv^{s}$ in the following.

When the eavesdropper is provided with $\fv^{s}$ in addition to already known $\ev$, she can remove the contribution of $\fv^{s}$ from the known $|\ev| = kp^k-p^k(k-\ell_1-\ell_2)\left(1-\frac{1}{p}\right)^{\ell_2}$ evaluations of the linearized polynomial
$$
\mathpzc{f}(y) = \sum_{i = 1}^{kp^k}f_iy^{q^{i-1}}
$$
to obtain $|\ev| = kp^k-p^k(k-\ell_1-\ell_2)\left(1-\frac{1}{p}\right)^{\ell_2}$ evaluations ($\tilde{\ev}$) of another linearized polynomial
$$
\tilde{\mathpzc{f}}(y) = \sum_{i = 1}^{|\ev|}r_iy^{q^{i-1}}
$$
at  $|\ev| = kp^k-p^k(k-\ell_1-\ell_2)\left(1-\frac{1}{p}\right)^{\ell_2}$  linearly independent (over $\F_q$) points in $\F_{q^m}$. From these evaluations of $\tilde{\mathpzc{f}}(\cdot)$, using Remark~\ref{rem:gabidulin_rec}, the eavesdropper can decode the random symbols $\rv$, which implies that $H(\rv|\fv^{s},\ev) = 0$. (We note that a similar proof of security when utilizing polynomials for encoding is provided in the seminal paper of A. Shamir on secret sharing~\cite{shamir79}.)

Substituting $\ell_2 = 1~(\text{or}~2),$ $\alpha = p^{k}$ and $\beta = \frac{p^{k}}{p} = p^{k-1}$ in (\ref{eq:sec_cap_gen_cor}) shows that the proposed code construction achieves the upper bound on secure file size, specified in Corollary~\ref{cor:intersec}, for $\ell_2 \leq 2$.
\end{proof}

\begin{figure*}
        \centering
        \begin{subfigure}[b]{0.4\textwidth}
                \centering
                \includegraphics[width=\columnwidth]{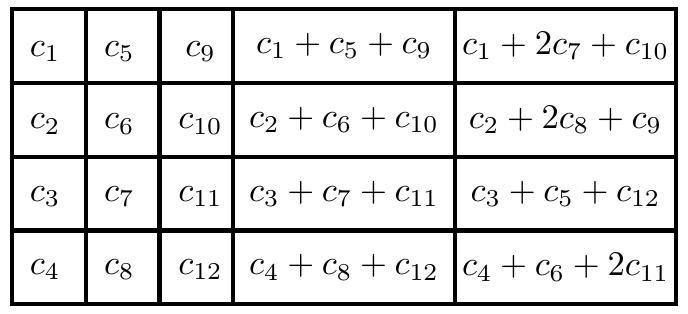}
                \caption{$(5,3)$-zigzag code.}
                \label{fig:zigzag_example}
        \end{subfigure}%
        \qquad
        \begin{subfigure}[b]{0.4\textwidth}
                \centering
                \includegraphics[width=\columnwidth]{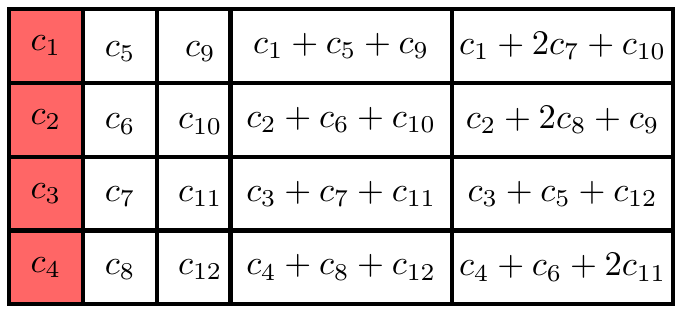}
                \caption{Observations of an eavesdropper with $\Ec_1 = \{1\}$.}
                \label{fig:zigzag_example2}
        \end{subfigure}\\
        \begin{subfigure}[b]{0.4\textwidth}
                \centering
                \includegraphics[width=\columnwidth]{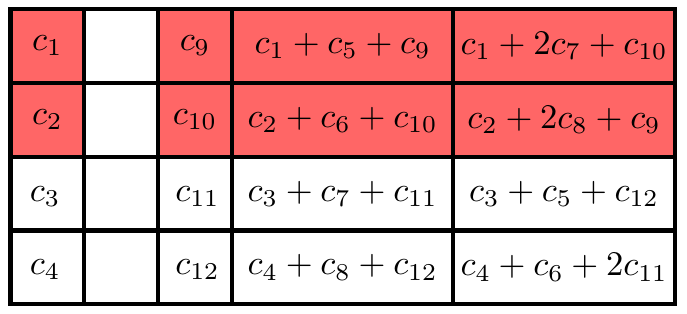}
                \caption{Encoded symbols obtained by an eavesdropper by observing repair of node $2$.}
                \label{fig:zigzag_example3}
        \end{subfigure}%
         \qquad
        \begin{subfigure}[b]{0.4\textwidth}
                \centering
                \includegraphics[width=\columnwidth]{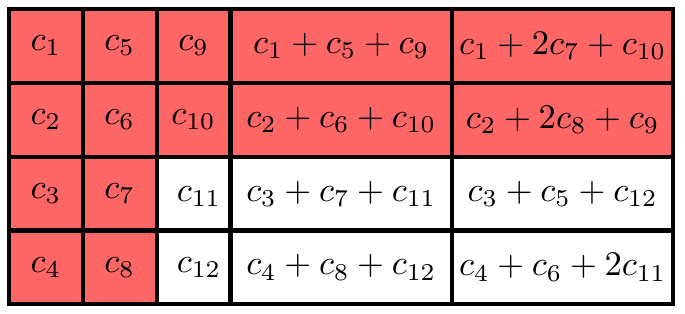}
                \caption{Observations of a $(1,1)$-eavesdropper with $\Ec_1 = \{1\}$ and $\Ec_2 = \{2\}$.}
                \label{fig:zigzag_example4}
        \end{subfigure}
        \caption{Illustrations of a $(5,3)$-DSS employing the proposed secure coding scheme. The output of Gabidulin pre-coding step $(c_1, c_2,\ldots, c_{12})$ is stored on $5$ node DSS using a $(5,3)$-zigzag code~\cite{zigzag13}. In this example, the eavesdropper is associated with $\Ec_1 = \{1\}$ and $\Ec_2 = \{2\}$.}\label{fig:zigzag_example}
\end{figure*}

\begin{remark}
\label{remark:goparaju}
In \cite{GRCP13}, Goparaju et al. obtain the following upper bound on the secrecy capacity at the MSR point, when $\Ec_2 \subset [k]$ and $d = n-1$,
\begin{align}
\label{eq:goparaju}
\cM^{s} \leq (k - \ell_1 - \ell_2)\left(1 - \frac{1}{n - k}\right)^{\ell_2}\alpha.
\end{align}
Therefore, this bound along with the coding scheme presented in this paper characterize the secrecy capacity at the MSR point, when $\Ec_2 \subset [k]$~{\cite[Theorem 3]{GRCP13}}.
\end{remark}

%

\begin{remark}
\label{remark:zigzag_allerton}
Zigzag codes, used in constructing secure coding scheme in this paper, do not allow bandwidth efficient repair of parity nodes. However, a related class of codes with $\alpha = (n - k)^{k + 1}$~\cite{zigzag_allerton11} allow bandwidth efficient repair for parity nodes as well. The repair set $Y_{j}$ for $j \in [k]$ (systematic nodes) for this class of codes have an intersection pattern similar to that of Zigzag codes. Therefore, Lemma~\ref{lm:union_zigzag} and consequently  Theorem~\ref{thm:msr_sec_cap} hold (with new value of $\alpha$) even when zigzag codes are replaced with the codes presented in~\cite{zigzag_allerton11}. The uses of codes from \cite{zigzag_allerton11} also give optimal coding schemes when  $|\Ec_2 \cap \{k+1,\ldots, k+p = n\}| \leq 1$. (The proof of this follows from the techniques developed in this section, and we omit the details for brevity.)
\end{remark}


\section{New Bounds and Constructions for Locally Repairable Codes}
\label{sec:lrc}

In this section, we study the notion of local-repairability for DSS. As opposed to the line of work on scalar locally repairable codes (LRCs) \cite{Gopalan12, pyramid, PKLK12}, where each node stores a scalar over a field from a codeword, we consider vector-LRCs~\cite{DimDim12, oggier_proj}. Our results can be specialized to scalar-LRCs by restricting node size $\alpha$ to $1$. Furthermore, the vector-LRCs that we consider, as defined in Section~\ref{sec:sys_model}, allow for the possibility of $\alpha> \Mc/k$, and non-trivial locality, i.e., $\delta> 2$. Thus, these codes are generalizations of the vector-LRCs studied in \cite{DimDim12} which considered only the $\delta=2$ case. We note that we are particularly interested in vector-LRCs with multiple local parities per local group.

We first derive an upper bound on the minimum distance of $(r, \delta, \alpha)$ LRCs, which also applies to non-linear codes. We follow the proof technique of \cite{Gopalan12, DimDim12}, which is given for the single local parity per local group case, by modifying it for multiple local parity nodes per local group. The bound derived in this section gives the bound presented in \cite{PKLK12} as a special case without the assumption of having a systematic code. As noted in \cite{DimDim12}, the bound on $d_{\min}$ establishes a resilience vs. per node storage trade-off, where per node storage $\alpha$ can be increased over $\Mc/k$ to obtain higher minimum distance. This is of particular interest in the design of codes having both local-repairability and strong resilience to node failures.

We then propose a general code construction for LRCs that achieve the derived bound on minimum distance of such codes. We use Gabidulin codes along with MDS array codes to obtain this construction.


\subsection{Upper bound on $d_{\min}$ for an $(r, \delta, \alpha)$ locally repairable code}
\label{subsec:d_min_upper}

Here, we state a generic upper bound on the minimum distance of an $(r, \delta, \alpha)$ LRC. We use the alternative definition of minimum distance $d_{\min}$ provided in Section~\ref{subsec:vec}.

\begin{theorem}\label{thm:dmin}
Let $\cC$ be an $(r, \delta, \alpha)$ LRC over $\mathbb{F}$ of length $n$ and cardinality $|\F|^{\cM}$. Then, the minimum distance of $\cC$ satisfies
\begin{align}
\label{eq:upp_bound}
&d_{\min}(\mathcal{C}) \leq n - \ceilb{\frac{\mathcal{M}}{\alpha}} + 1 - \left(\ceilb{\frac{\mathcal{M}}{r\alpha}}-1\right)(\delta - 1).
\end{align}
\end{theorem}

\begin{proof}
In order to get  this upper bound on the minimum distance of an LRC, we utilize the dual definition of minimum distance of a code as given in \eqref{eq:def_dmin}. Similar to the proof in \cite{Gopalan12} and \cite{DimDim12}, we construct a set $\Ac \subseteq [n]$ such that
\begin{equation}
H(\cv_{\Ac}) < \mathcal{M}.
\end{equation}
This along with \eqref{eq:def_dmin} in Definition~\ref{def:dmin} give us the desired upper bound on $d_{\min}(\mathcal{C})$. Here, $\cv \in \Cc$ denote a codeword selected uniformly at random from $\F^\Mc$ codewords in $\cC$.

\begin{figure}
\algrule[3pt]
\begin{algorithmic}[1]
\STATE Set $\Ac_{0} = \emptyset$ and $i = 1$.
\WHILE{$H(\cv_{\Ac_{i-1}}) < \mathcal{M}$}
\STATE Pick a coded block $\cv_{j_i} \notin \Ac_{i-1}$ s.t. $|\Gamma(j_i)\backslash \Ac_{i-1}| \geq \delta-1$.
\IF{$H(\cv_{\Ac_{i-1}}, \cv_{\Gamma(j_i)}) < \mathcal{M}$}
\STATE set $\Ac_{i} = \Ac_{i-1}\cup \Gamma(j_i)$
\ELSIF{$H(\cv_{\Ac_{i-1}}, \cv_{\Gamma(j_i)}) \geq \mathcal{M}$ and $\exists \Bc \subset \Gamma(j_i)$ s.t. $H(\cv_{\Ac_{i-1}}, \cv_{\Bc}) < \mathcal{M}$}
\STATE set $\Ac_{i} = \Ac_{i-1} \cup \Bc$
\ELSE
\STATE $i = i+1$, \textbf{end while}
\ENDIF
\STATE $i = i+1$
\ENDWHILE
\STATE Output: $\Ac = \Ac_{i-1}$
\end{algorithmic}
\algrule[3pt]
\caption{Construction of a set $\Ac$ with $H(\cv_{\Ac}) < \mathcal{M}$ for an $(r, \delta, \alpha)$ LRC.}
\label{fig:construction_set}
\end{figure}

The algorithm of construction of set $\Ac$ is given in Fig.~\ref{fig:construction_set}.
Next, we show a lower bound on the size of the set $\Ac$, output of the algorithm described in Fig.~\ref{fig:construction_set}. Note that at each iteration of the while loop in Fig.~\ref{fig:construction_set}, the algorithm increases the size of the set $\Ac_{i-1}$ by at most $r+\delta - 1$ to get $\Ac_{i}$. For each $i$, define
\begin{align}
\label{eq:size}
a_{i} &= |\Ac_i| - |\Ac_{i-1}|.
\end{align}
and
\begin{align}
\label{eq:entropy}
h_{i} &= H(\cv_{\Ac_i}) - H(\cv_{\Ac_{i-1}}).
\end{align}

Assume that the algorithm terminates at $(\ell+1)^{th}$ iteration, i.e., $\Ac = \Ac_{\ell}$. Then, it follows from (\ref{eq:size}) and (\ref{eq:entropy}) that
\begin{align}
\label{eq:sum_size}
|\Ac| &= |\Ac_{\ell}| = \sum_{i=1}^{\ell}a_{i},
\end{align}
and,
\begin{align}
\label{eq:sum_ent}
H(\cv_{\Ac}) &= H(\cv_{\Ac_{\ell}}) = \sum_{i = 1}^{\ell}h_{i}.
\end{align}
Consider two cases depending on the way the algorithm in Fig.~\ref{fig:construction_set} terminates:

\textbf{Case 1:}
Assume that the algorithm terminates with the final set assigned at step $5$, i.e., after adding $\Gamma(j_\ell)$ to $\Ac_{\ell - 1}$. Now we have from $(r, \delta)$ locality of the code that
\begin{eqnarray}
\label{eq:ent_size}
h_i &=& H(\cv_{\Ac_i}) - H(\sv_{\Ac_{i-1}}) \nonumber \\
&=& H(\cv_{\Ac_{i-1}\cup (\Ac_{i}\backslash \Ac_{i-1})}) - H(\cv_{\Ac_{i-1}}) \nonumber \\
&=& H(\cv_{\Ac_{i-1}}) + H(\cv_{\Ac_{i}\backslash \Ac_{i-1}}| \cv_{\Ac_{i-1}}) - H(\cv_{\Ac_{i-1}}) \nonumber \\
&=& H(\cv_{\Ac_{i}\backslash \Ac_{i-1}}| \cv_{\Ac_{i-1}}) \nonumber \\
&\leq& (a_{i} - \delta + 1)\alpha.
\end{eqnarray}
The last inequality follows from the fact that any block in $\Gamma(j_i)$ can be written as a function of any set of $r$ blocks in $\Gamma(j_i)$ and the fact that we pick $i$ in step $3$ only if $|\Gamma(j_i)\backslash \Ac_{i-1}| \geq \delta-1$. Since at the end of $i^{th}$ iteration, we have all the elements of $\Gamma(j_i)$ added to $\Ac_{i}$, out of which $a_{i}$ blocks are added at the $i^{th}$ iteration. These newly added packets can not contribute more than $(a_i-(\delta-1))\alpha$ to the entropy of set $\Ac_{i}$ as $\delta-1$ of these packets are deterministic function of other newly added blocks of $\Gamma(j_i)$ and blocks of $\Gamma(j_i)$ that were already present in $\Ac_{i-1}$.
From (\ref{eq:ent_size}), we have that
\begin{equation}
\label{eq:ent_size2}
a_{i} \geq \frac{h_i}{\alpha} + \delta -1.
\end{equation}
Now using (\ref{eq:sum_size})
\begin{eqnarray}
\label{eq:der_1}
|\Ac| = |\Ac_{\ell}| &=& \sum_{i=1}^{\ell}a_{i} \nonumber \\
&\geq& \sum_{i=1}^{\ell}\left(\frac{h_i}{\alpha} + \delta -1\right) \nonumber \\
&=& \frac{1}{\alpha}\sum_{i=1}^{\ell}h_i + (\delta - 1)\ell. \nonumber \\
\end{eqnarray}
Similar to the proof of Papailiopoulos et al. \cite{DimDim12}, we have
\begin{equation}
\label{eq:ent_eq}
\sum_{i=1}^{\ell}h_i  = \left(\ceilb{\frac{\mathcal{M}}{\alpha}}\alpha -\alpha \right),
\end{equation}
and
\begin{equation}
\label{eq:l_bound}
\ell = \ceilb{\frac{\mathcal{M}}{r\alpha}} - 1.
\end{equation}
It follows from (\ref{eq:der_1}), (\ref{eq:ent_eq}), and (\ref{eq:l_bound}) that
\begin{equation}
\label{eq:upp_bound1}
|\Ac_{\ell}| \geq \ceilb{\frac{\mathcal{M}}{\alpha}} - 1 + \left(\ceilb{\frac{\mathcal{M}}{r\alpha}} - 1\right)(\delta - 1).
\end{equation}

\textbf{Case 2:}
The proof of this case is exactly similar to that in \cite{Gopalan12} except a few minor modification.
Consider that the algorithm terminates with the final set assigned at step $7$ in $\ell^{th}$ iteration. Since it reaches the step $7$, we have
\begin{equation}
H(\cv_{\Ac_{\ell-1}\cup \Gamma(j_{\ell})}) \geq \mathcal{M}.
\end{equation}
As the increment in the entropy is at most $r\alpha$ at each iteration, we have
\begin{equation}
\label{eq:l_bound2}
\ell \geq \ceilb{\frac{\mathcal{M}}{r\alpha}}
\end{equation}
For $i \leq \ell-1$, from (\ref{eq:ent_size2})
\begin{equation}
a_{i} \geq \frac{h_i}{\alpha} + \delta -1.\label{eq:l_bound3_1}
\end{equation}
For $i = \ell$,
\begin{equation}
\label{eq:l_bound3_2}
a_{\ell} \geq \frac{h_{\ell}}{\alpha}.
\end{equation}
Next, it follows from (\ref{eq:sum_size}), (\ref{eq:l_bound2}), (\ref{eq:l_bound3_1}), and (\ref{eq:l_bound3_2}) that
\begin{eqnarray}
\label{eq:upp_bound2}
|\Ac_{\ell}| &=& \sum_{i = 1}^{\ell}a_{i} \nonumber \\
&\geq& \sum_{i = 1}^{\ell-1}\left(\frac{h_i}{\alpha} + \delta -1\right) + \frac{h_{\ell}}{\alpha} \nonumber \\
&=& \frac{1}{\alpha}\sum_{i=1}^{\ell}h_i + (\ell-1)(\delta - 1) \nonumber \\
&\geq& \frac{1}{\alpha}\left(\ceilb{\frac{\mathcal{M}}{\alpha}}\alpha -\alpha \right) +  \left(\ceilb{\frac{\mathcal{M}}{r\alpha}}-1\right)(\delta-1) \\
&=& \ceilb{\frac{\mathcal{M}}{\alpha}} - 1 + \left(\ceilb{\frac{\mathcal{M}}{r\alpha}} - 1\right)(\delta - 1) \label{eq:upp_bound22}
\end{eqnarray}
where (\ref{eq:upp_bound2}) follows from (\ref{eq:l_bound2}) and (\ref{eq:ent_eq}). Now combining \eqref{eq:def_dmin}, (\ref{eq:upp_bound1}), and (\ref{eq:upp_bound22}), we get
\begin{equation}
d_{\min}(\mathcal{C}) \leq n - \ceilb{\frac{\mathcal{M}}{\alpha}} + 1 - \left(\ceilb{\frac{\mathcal{M}}{r\alpha}} - 1\right)(\delta - 1).
\end{equation}
\end{proof}

For the special case of $\delta=2$, this bound matches with the bound obtained in~\cite{DimDim12}. For the case of $\Mc = k$ and $\alpha=1$, the bound reduces to $d_{\min}(\cC) \leq n-k+1+(\ceilb{k/r}-1) (\delta-1)$, which is coincident with the bound presented in~\cite{PKLK12}.


\subsection{Construction of $d_{\min}$-optimal locally repairable codes}
\label{subsec:optimal_local_repairable}
In this subsection, we present a construction of an $(r,\delta,\alpha)$ LRC, which attains the bound given in Theorem~\ref{thm:dmin}. Given a file (over $\F_{q^m}$) to be stored on DSS, we encode the file in two steps before placing it on DSS. First, the file is encoded using a Gabidulin code. The codeword  (over  $\F_{q^m}$) of the Gabidulin code is partitioned into local groups; then, each of these local groups is encoded using an MDS array code over $\mathbb{F}_q$ (i.e., its generator matrix is defined over $\F_q$). This construction can be viewed as a generalization of the construction for scalar LRCs proposed in~\cite{SRV12}.

\textbf{Construction I.} Consider a file $\fv$ over $\F=\F_{q^m}$ of size $\cM \geq r\alpha$, where $m$ will be defined in the sequel. In addition, let $n$,~$r$,~and~$\delta$~be positive integers such that $r+\delta-1<n$. We use $g=\left\lceil\frac{n}{r+\delta-1}\right\rceil$ to denote the number of local groups in an LRC obtained by the construction proposed here. Depending on whether $(r + \delta - 1)$ divides $n$ or not, the details of the construction are as follows:

\begin{itemize}
  \item \textbf{Case 1 $\left((r+\delta-1)|n\right)$:}~ Let $N=\frac{nr\alpha}{r+\delta-1}$, $m\geq N$, and $\cC^{\rm{Gab}}$ be an $[N, K = \cM, D = N-\cM+1]_{q^m}$ Gabidulin  code. The encoding process of file $\fv$ consists of two steps:
\begin{enumerate}
\item First, encode $\cM$ symbols of $\fv$ to a codeword $\cv \in \cC^{\rm{Gab}}$, and partition $\cv$ into $g=\frac{N}{r\alpha}$ disjoint groups, each of size $r\alpha$. Each group is then stored on a different set of $r$ nodes, $\alpha$ symbols per node.
\item Next, generate $\delta-1$ parity nodes per group by applying an $[(r+\delta-1),r\alpha,\delta,\alpha]_q$ MDS array code
      on each local group of $r$ nodes, treating these $r$ nodes as input data blocks (of length $\alpha$) for the MDS array code.
\end{enumerate}
 At the end of this encoding process, we have $n=g(r+\delta-1)=\frac{N}{\alpha}+\frac{N}{r\alpha}(\delta-1)$ nodes, each storing $\alpha$ symbols over $\mathbb{F}_{q^{m}}$; and $g$ local groups, each of size $r-\delta+1$.

  \item  \textbf{Case 2 $\left(n~({\rm mod}~r+\delta-1)-(\delta-1)>0\right)$:}~Let $\beta_0$, $1\leq\beta_0\leq r-1$, be an integer such that $n=\lfloor\frac{n}{r+\delta-1}\rfloor(r+\delta-1)+\beta_0+\delta-1=(g-1)(r+\delta-1)+\beta_0+\delta-1$. Let $N=(g-1)r\alpha+\beta_0\alpha$, $m\geq N$, and $\cC^{\rm{Gab}}$ be an $[N, \cM, D = N-\cM+1]_{q^m}$ Gabidulin code. Two step encoding process for $\fv$ is described as follows:
\begin{enumerate}
\item  First, encode $\cM$ symbols of $\fv$ to a codeword $\cv \in \cC^{\rm{Gab}}$, and partition $\cv$ into $g-1$ disjoint groups of size $r\alpha$ and one additional group of size $\beta_0\alpha$. First $g-1$ groups are stored on $r(g-1)$ nodes, and the last group is stored on $\beta_0$ nodes, each node storing $\alpha$ symbols.
\item Next, generate $\delta-1$ parity nodes per group by applying an $[(r+\delta-1),r\alpha,\delta,\alpha]_q$ MDS array code   on each of the first $g-1$ local groups of $r$ nodes, and by applying a $[(\beta_0+\delta-1),\beta_0\alpha,\delta,\alpha]_q$ MDS array code  on the last local group.
\end{enumerate}
 At the end of the second step of encoding, we have $n=(g-1)(r+\delta-1)+(\beta_0+\delta-1)=\frac{N}{\alpha}+\left\lceil\frac{N}{r\alpha}\right\rceil(\delta-1)$ nodes, each storing $\alpha$ symbols over $\mathbb{F}_{q^{m}}$; and $g$ local groups, $g-1$ of which have size $r-\delta+1$ and one group of size $\beta_0+\delta-1$.
\end{itemize}
The output of the Construction I is denoted by $\cC^{\rm loc}$.
\begin{remark}
Note that the output of the first encoding step generates the encoded data stored on $rg$  and $r(g -1) + \beta_0$ nodes in case $1$ and case $2$, respectively. Each of these nodes stores $\alpha$ symbols of a ({\em folded}) Gabidulin codeword.
\end{remark}

Before proceeding to show the optimality of Construction I with respect to the minimum distance bound in \eqref{eq:upp_bound}, we present the following result which we use later in this paper.

\begin{lemma}
\label{lem:linearized_propertyI}
Let $s_i$, for $i \in [g]$, denotes the number of nodes observed in $i$-th local group of a DSS employing $\Cc^{\rm loc}$. Then, the data stored on the observed nodes corresponds to the evaluations of the underlying linearized polynomial $\mathpzc{f}(\cdot)$ at $\sum_{i = 1}^{g}\min\{s_i, r\}\alpha$ and $\sum_{i = 1}^{g-1}\min\{s_i, r\}\alpha + \min\{s_g, \beta_0\}$ linearly independent (over $\F_q$) points from $\F_{q^m}$ in case 1 and case 2 of Construction I, respectively.
\end{lemma}
\begin{proof}
Here, we present the proof only for case $1$, i.e., $(r + \delta - 1) | n$. The claim for case $2$ can be proved using the similar steps.
Note that the second step of Construction I encodes evaluations of a linearized polynomial at different sets of linearly independent (over $\F_q$) points using MDS array codes (defined over $\F_q$); therefore, Lemma~\ref{lem:linearized_property} is applicable.

Let $\{\mathpzc{f}(y_{(i-1)r\alpha + 1}),\ldots, \mathpzc{f}(y_{i r\alpha})\}$ denote the symbols of a Gabidulin codeword that are encoded using an MDS array code in $i$-th local group, for $i \in [g]$. It follows from Lemma~\ref{lem:linearized_property} that {\em any} $s \leq r$ nodes (vector symbols) inside a local group of $\Cc^{\rm loc}$ correspond to evaluations of the underlying linearized polynomial at $s\alpha$ linearly independent (over $\F_q$) points in the subspace spanned by $\yv_i = \{y_{(i-1)r\alpha + 1},\ldots, y_{i r\alpha}\}$. For $s>r$, on the other hand, any $s$ vector symbols of a local group can only have $r\alpha$ linearly independent points (over $\F_q$). Since points in $\yv = \{\yv_1,\ldots, \yv_g\}$ are linearly independent, i.e., the evaluation points of one group are linearly independent on the evaluation points of any other group, we have that the data obtained from $\sum_{i = 1}^{g}s_i$ nodes, $s_i$ nodes from $i$-th local group, corresponds to the evaluations of $\mathpzc{f}(\cdot)$ at $\sum_{i = 1}^{g}\min\{s_i, r\}\alpha$ linearly independent (over $\F_q$) points in $\F_{q^m}$.
\end{proof}

Lemma~\ref{lem:linearized_propertyI} also implies that \emph{any} $\delta- 1 +  i$ \emph{node erasures} in a local group correspond to $i\alpha$ \emph{rank erasures} (in the corresponding Gabidulin codeword). Next, we provide the conditions for parameters of the code $\Cc^{\rm loc}$ obtained from Construction I to be a $d_{\min}$-optimal $(r,\delta,\alpha)$ LRC.
\begin{theorem}
\label{thm:parameters}

Let $\Cc^{\rm loc}$ be an $(r, \delta, \alpha)$ LRC of length $n$ and cardinality $|\F|^\cM$ obtained by Construction I. Then,
\begin{itemize}
  \item If $ (r+\delta-1)|n$, then $\Cc^{\rm loc}$ over $\F=\F_{q^m}$, for $m\geq \frac{n r \alpha}{r+\delta-1}$ and $q\geq (r+\delta-1)$, attains the bound~(\ref{eq:upp_bound}).
  \item If $n(\textmd{mod } r+\delta-1)-(\delta-1)\geq
 \left\lceil\frac{\cM}{\alpha}\right\rceil(\textmd{mod } r)>0$,
  then $\Cc^{\rm loc}$ over $\F=\F_{q^m}$, for $m\geq \alpha\left(n-(\delta-1)\left(\left\lfloor\frac{n}{r+\delta-1}\right\rfloor+1\right)\right)$ and $q\geq (r+\delta-1)$, attains the bound~(\ref{eq:upp_bound}).
\end{itemize}
\end{theorem}

\begin{proof}
The proof  is based on Lemma~\ref{lem:linearized_propertyI}. We note that any
 $n - \ceilb{\frac{\mathcal{M}}{\alpha}} - \left(\ceilb{\frac{\mathcal{M}}{r\alpha}}-1\right)(\delta - 1)$ node erasures correspond to at most $D-1$ rank erasures (erasures in the corresponding Gabidulin codeword) which can be corrected by the Gabidulin code $\cC^{\rm{Gab}}$.  See the details in Appendix~\ref{ap:appendix}.
  \end{proof}

  Specializing the Construction I to scalar case ($\alpha = 1$), we obtain explicit $(r, \delta)$ scalar LRCs for the settings of parameters, where only results on existence of scalar LRCs are present in literature. For scalar case, Construction~I employs MDS codes instead of MDS array codes after encoding information symbols $\fv$ with Gabidulin codes. The following corollary (of Theorem~\ref{thm:parameters}) summarizes our contributions towards designing scalar-LRCs.

\begin{corollary}
\label{cor.parameters_scalar}

Let $\Cc^{\rm loc}$ be a $(r, \delta)$ scalar LRC  obtained by Construction I. Then,
\begin{itemize}
  \item If $ (r+\delta-1)|n$, then $\Cc^{\rm loc}$ over $\F=\F_{q^m}$, for $m\geq \frac{nr}{r+\delta-1}$ and $q\geq (r+\delta-1)$, attains the bound~(\ref{scalarUpperBound}).
  \item If $n(\textmd{mod } r+\delta-1)-(\delta-1)\geq \cM (\textmd{mod } r) > 0$,
  then $\Cc^{\rm loc}$ over $\F=\F_{q^m}$, for $m\geq \left(n-(\delta-1)\left(\left\lfloor\frac{n}{r+\delta-1}\right\rfloor+1\right)\right)$ and $q\geq (r+\delta-1)$, attains the bound~(\ref{scalarUpperBound}).
\end{itemize}
\end{corollary}

%

\begin{remark}
The required field size  $|\F|=q^m$ for the proposed construction should satisfy $m\geq N$, given a choice of $q \geq (r+\delta-1)$, where $N=\frac{nr}{r+\delta-1}$ for the first case and $N=\left(n-(\delta-1)\left(\left\lfloor\frac{n}{r+\delta-1}\right\rfloor+1\right)\right)$ for the second case. So, we can assume that $|\F|=q^N$. Here, the requirement of $q \geq (r + \delta - 1)$ is enforced as it is a sufficient condition for having an $(r + \delta - 1, r)$ MDS code.
\end{remark}

We illustrate the construction of $\Cc^{\rm loc}$ in the following examples. First, we consider the scalar case.

\begin{example}
Consider the following system parameters:
$$(\cM,n,r,\delta,\alpha)=(9,14,4,2,1).
$$
Here, we have $r+\delta-1 = 5$ and $14~({\rm mod}~5)-1=3$; therefore, the given choice of parameters corresponds to the second case of Construction I. Since $n=\left\lfloor\frac{14}{4+2-1}\right\rfloor\cdot (4+2-1)+(3+2-1)$, let $N=\left\lfloor\frac{14}{4+2-1}\right\rfloor \cdot4+3=11$. First, $\cM=9$ symbols over $\F=\F_{5^{11}}$  are encoded into a codeword $\cv$ of an $[11,9,3]_{5^{11}}$ Gabidulin code $\cC^{\rm{Gab}}$. This codeword is partitioned into three groups, two of size $4$ and one of size $3$, as follows: $\cv=(a_1,a_2,a_3,a_4| b_1,b_2,b_3,b_4|c_1,c_2,c_3)$. Then, we add one parity to each group by applying a $[5,4,2]$ MDS code over $\F_5$ in the first two groups and a $[4,3,2]$ MDS code over $\F_5$ in the last group. The symbols of $\cv$ with three new parities $p_a,p_b,p_c$ are stored on 14 nodes as shown in Fig~\ref{fig:construction1}.
\begin{figure}[t]
 \centering
 \includegraphics[width=1\columnwidth]{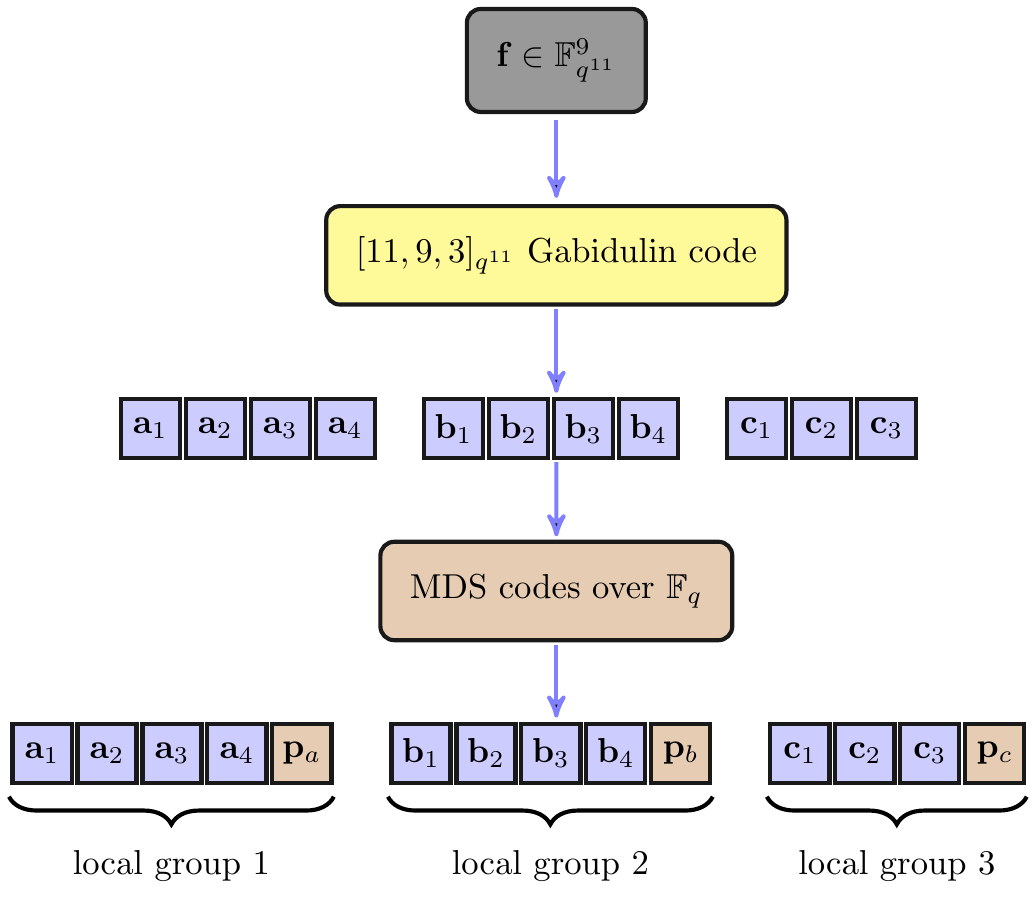}
 \caption{Illustration of the construction of a scalar ${(r = 4, \delta = 2, \alpha = 1)}$ LRC for $n = 14, \cM=9$ and ${d_{\min}=4}$.} \label{fig:construction1}
\end{figure}

It follows from Theorem~\ref{thm:dmin} that the minimum distance $d_{\min}$ of this code is at most $4$. To figure out the minimum distance of the code, we first observe from Lemma~\ref{lem:linearized_propertyI} that any $3$ node erasures translate into at most $2$ rank erasures. Then, observing that these rank erasures can be corrected by $\cC^{\rm{Gab}}$ (as it has $D=3$ here), we obtain that $d_{\min}=4$. In addition, when a single node fails, it can be repaired by using the data stored on all the  other nodes from the same local group.
\end{example}

Next, we describe Construction I for vector-LRCs with the help of the following example.

\begin{example}
\label{ex:array}
We consider a DSS with the following parameters:
$$(\cM,n,r,\delta,\alpha)=(28,15,3,3,4).
$$
First, $r+delta-1=5$, then $r+\delta-1|n$, and we consider the first case in Construction~I.
Let $N=\frac{15\cdot 3\cdot 4}{3+3-1}=36$ and $(a_1,\ldots, a_{12}, b_1,\ldots, b_{12}, c_1,\ldots, c_{12})$ be a codeword  of a $[36,28,9]_{q^{36}}$ code $\cC^{\textmd{Gab}}$, which is obtained by encoding $\cM = 28$ symbols over $\F=\F_{q^{36}}$ of the original file. The Gabidulin codeword is then partitioned into three groups $(a_1,\ldots, a_{12})$, $(b_1,\ldots, b_{12})$, and $(c_1,\ldots, c_{12})$. Encoded symbols from each group are stored on three storage nodes as shown in Fig.~\ref{fig:construction}. In the second stage of encoding, a $[5,12,3,4]_{q}$  MDS array code over $\F_{q}$  is applied on each local group to obtain $\delta-1 = 2$ parity nodes per local group. The coding scheme is illustrated in Fig.~\ref{fig:construction}.

%

It follows from Theorem~\ref{thm:dmin} that the minimum distance of an LRC with the parameters of this example is at most~$5$. Using Lemma~\ref{lem:linearized_propertyI}, we conclude that any $4$ node failures correspond to at most $8$ rank erasures in the corresponding codeword of $\cC^{\textmd{Gab}}$. Since the minimum rank distance of $\cC^{Gab}$ is $9$, these node erasures can be corrected by $\cC^{Gab}$; thus, the minimum distance of $\Cc^{\rm loc}$ is exactly $5$.
\end{example}

\begin{figure}[h]
 \centering
 \includegraphics[width=1\columnwidth]{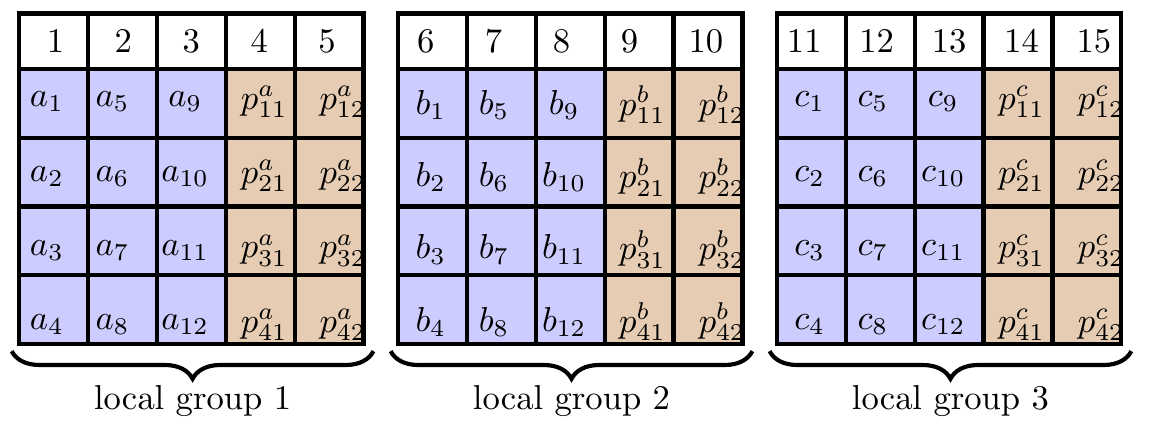}
 \caption{Example of an $(r = 3, \delta = 3, \alpha = 4 )$ LRC with $n = 15$, $\cM=28$ and $d_{\min}=5$.} \label{fig:construction}
\end{figure}

%
%


\section{Repair Bandwidth Efficient Locally Repairable Codes}
\label{sec:local_repair_bw}

In this section, we introduce the notion of repair bandwidth for LRCs, and present hybrid codes that allow for local repairs while minimizing repair bandwidth for given locality parameters. As pointed out in Section~\ref{subsec:lrc}, LRCs allow for na\"{\i}ve repair process, where a newcomer contacts $r$ nodes in its local group and downloads all the data stored on these $r$ nodes. The newcomer then regenerates the data stored on the failed node and stores it for future operations. In this section, following the line of work on repair bandwidth efficient DSS due to \cite{dimakis}, we consider allowing a newcomer to contact more than $r$ nodes in its local group in order to repair the failed node. The motivation behind this is to lower the repair bandwidth of an LRC. (This also improves the secrecy capacity of such codes as detailed in Section~\ref{sec:secrecy_local}.)  The main idea here is to apply a regenerating code in each local group. (We note that, in a parallel and independent work, Kamath et al.~\cite{KPLV12} also proposed utilizing regenerating codes to perform efficient local repairs.).

First, we provide an upper bound on the amount of data that can be stored on a locally repairable DSS while supporting a given repair bandwidth $d\beta$ and the maximum possible failure resilience (i.e., maximum possible minimum distance). We then  present $d_{\min}$-optimal LRCs which attain this bound by applying an MSR code in each local group instead of any MDS array code in the second step of Construction I. We denote such codes by MSR-LRC.

\subsection{File size upper bound for repair bandwidth efficient LRCs}
\label{subsec:local_repair_bw}

In the rest of this section, we restrict ourselves to LRCs that have the maximum possible minimum distance as described in (\ref{eq:upp_bound}). Moreover, we restrict ourselves to $(r + \delta -1)|n$ case. We remark that, for $(r+\delta-1)|n$, the upper bound on minimum distance for LRCs given in (\ref{eq:upp_bound}) is achievable only if the code have disjoint local groups~\cite{KPLV12}. Let $\Gc_1,\ldots, \Gc_{g}$ denote $g=\frac{n}{r+\delta-1}$ disjoint sets of indices of storage nodes with $\Gc_i$ representing $(r + \delta - 1)$ nodes in $i$-th local group. A failed node in a particular local group is repaired by contacting $d$ remaining nodes within the group, where $r\leq d \leq r+\delta-2$. A newcomer downloads $\beta$ symbols from each of these $d$ nodes during the node repair process.

Next, we perform the standard max-flow min-cut based analysis for locally repairable DSS by mapping it to a multicasting problem on a dynamic information flow graph. The information flow graph representation of a locally repairable DSS is obtained by modifying the information flow graph for classical DSS~\cite{dimakis}. This is first introduced in \cite{DimDim12} for na\"{\i}ve repair, where the newcomer contacts $r$ nodes during node repair. We assume a sequence of node failures and node repairs as shown in Fig.~\ref{fig:flow_graph}. We further assume that each local group encounters the same sequence of node failures and node repairs that are performed as result of these failures. Each data collector contacts $n-d_{\min}+1$ storage nodes for data reconstruction. A data collector is associated with the nodes it contacts for data reconstruction, and we represent the data collector by $(\Kc_1, \Kc_2,\ldots, \Kc_g)$. Here, $\Kc_i \subseteq \Gc_i$ is the set of indices of nodes that the data collector contacts in $i$-th local group and $\sum_{i=1}^{g}|\Kc_i| = n-d_{\min}+1$. Next, we derive an upper bound on the amount of data that can be stored on an $(n, n-d_{\min}+1)$-DSS employing an $(r, \delta, \alpha)$ LRC. Note that an $(n, n-d_{\min}+1)$-DSS allows a data collector to recover the original file from the data stored on any set of $n-d_{\min}+1$ nodes. This upper bound is used to derive a repair bandwidth vs. per node storage trade-off for minimum distance optimal $(r,\delta,\alpha)$ LRCs. In what follows, we add two more parameters in the representation of LRCs and denote them by the tuple $(r,\delta,\alpha, \beta, d)$. These parameters along with $n$  and  $\Mc$, the number of nodes and the file size, characterize a bandwidth efficient $d_{\min}$-optimal LRC.

\begin{figure*}[t]
 \centering
 \includegraphics[width=1\textwidth]{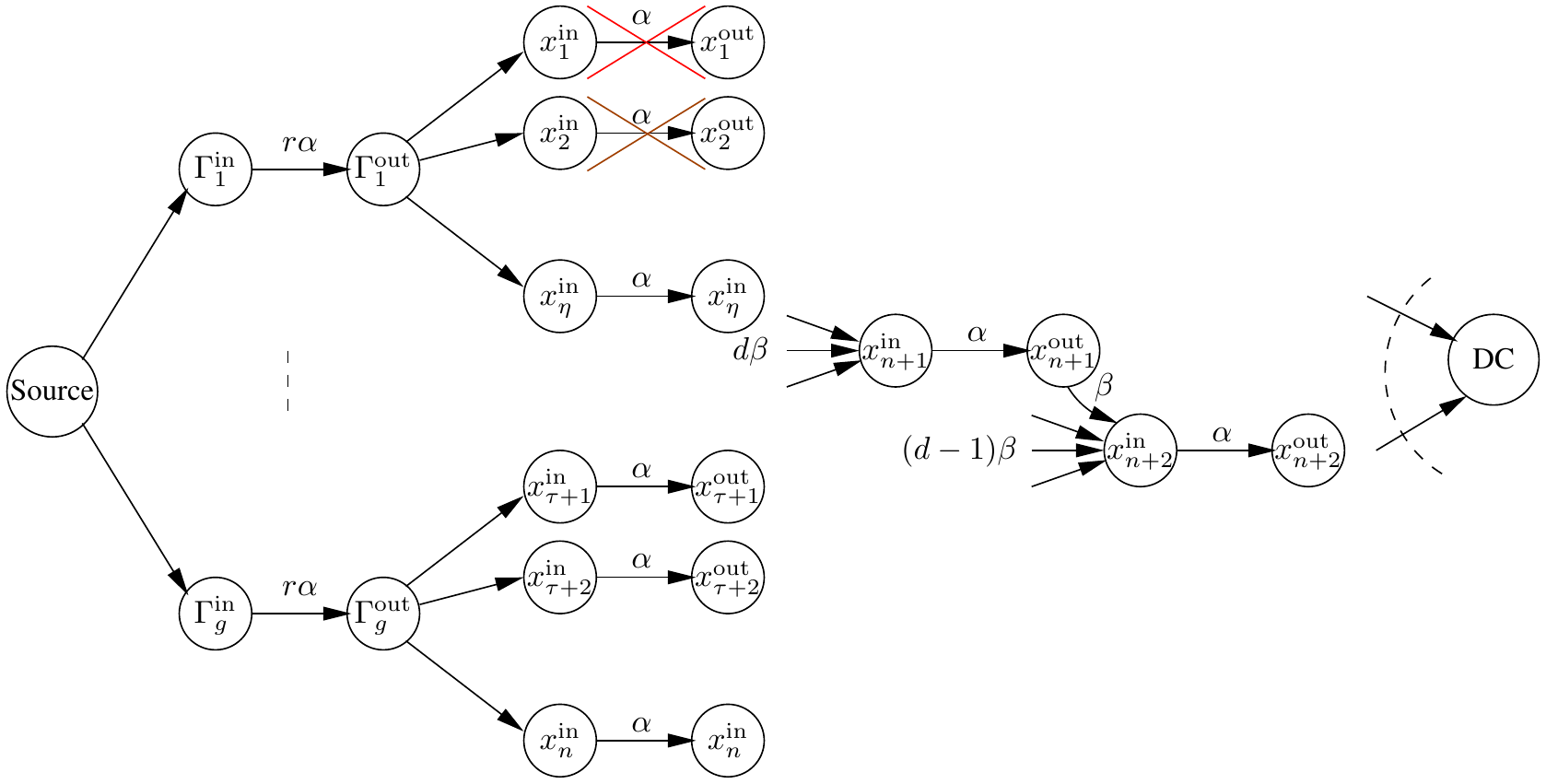}
 \caption{Flow graph for $(r,\delta,\alpha, \beta, d)$ LRC. In this graph, node pairs $\{\Gamma^{\rm in}_i,\Gamma^{\rm out}_i\}_{i=1}^{g}$ with edge of capacity $r\alpha$ enforce the requirement that each local group has at most $r\alpha$ entropy. Here $\eta$ and $\tau$ denote $r +\delta -1$ and $n - (r + \delta - 1)(g-1)$ respectively. The figure illustrates the case where  $x_1$ and $x_2$ sequentially fail and are replaced by introducing $x_{n+1}$ and $x_{n+2}$ respectively. The data collector (DC) is assumed to contact a set of $n-d_{\min}+1$ nodes for reconstruction of original data.}\label{fig:flow_graph}
\end{figure*}

\begin{theorem}
\label{thm:bw_locally}
For an $n$-node  DSS employing an $(r,\delta,\alpha,\beta, d)$ LRC to store a file of size $\Mc$, we have

\begin{align}
\label{eq:bw_upper}
&\mathcal{M} \leq
 \min \Bigg\{r\alpha, \sum_{i = 0}^{h-1}\min\{\max\{(d-i)\beta,0\}, \alpha\}\Bigg\} \\ \n
&\quad\quad\:+\:\sum_{j=1}^{\mu}\min \Bigg\{r\alpha, \sum_{i = 0}^{r+\delta-2}\min\{\max\{(d-i)\beta,0\}, \alpha\}\Bigg\},
\end{align}
where $\mu = \floorb{\frac{n-d_{min}+1}{r+\delta-1}}$ and $h = n-d_{\min}+1 - (r +\delta -1)\floorb{\frac{n-d_{\min}+1}{r+\delta-1}}$.
\end{theorem}

\begin{proof}
Consider a data collector with $\Kc_1 = \Gc_1, \Kc_2 = \Gc_2,\ldots, \Kc_{\mu} = \Gc_{\mu}, \Kc_{\mu+2} = \ldots = \Kc_{g} = \emptyset$, and $\Kc_{\mu+1} \subset \Gc_{\mu + 1}$ such that $\left|\Kc_{\mu+1}\right| = h$. Now, the bound in (\ref{eq:bw_upper}) follows by finding various cuts in the information flow graph (Fig.~\ref{fig:flow_graph}). For each group, we consider cuts similar to the ones given in~\cite{dimakis}. Here, the data collector connects to $h$ nodes for the first term in \eqref{eq:bw_upper} and $r+\delta-1$ nodes for each of the terms in the summation of the second term in \eqref{eq:bw_upper}. Now, consider the $i$-th node out of $\tilde{k}$ nodes that data collector connects in a particular group. (Here, $\tilde{k}=h$, or $\tilde{k}=r+\delta-1$ as described above.) A cut between $x_i^{\rm in}$ and $x_i^{\rm out}$ for each node gives a cut-value of~$\alpha$. On the other hand, for $i=0,\cdots,\tilde{k}-1$, if the cut is such that $x_i^{\rm in}$ belongs to the data collector side, we consider that $(d-i)$ live nodes are connected together with $i$ nodes that have been previously repaired. In our setup, for such a cut, the cut-value evaluates to $\max\{(d-i)\beta,0\}$. For $i>d$, the repair node is considered to contact only the previously repaired nodes, and hence does not contribute to the maximum flow.
\end{proof}

Note that the codes that are under consideration have property that each local group has entropy of $r\alpha$ and any set of $r$ nodes has $r\alpha$ independent symbols. (See definition of $(r, \delta, \alpha)$ LRC in Section~\ref{subsec:lrc}.) Therefore, node repairs within each local group have to ensure this property. This implies that each local group and its repair can be related to an \textbf{$(r+\delta-1, r, \delta, \alpha, \beta, d)$} MSR code with a file of size $r\alpha$. Hence, when a data collector connects to any $r$ nodes in a group, it can get all the information that particular group has to offer. Therefore, similar to the analysis given in~\cite{dimakis} for the classical setup, the parameters need to satisfy
\begin{eqnarray}
\label{eq:msr_bound}
r\alpha = \sum_{i = 0}^{r-1}\min\{(d-i)\beta,\alpha\},
\end{eqnarray}
which leads to the requirement of $(d-i)\beta\geq \alpha$ for each $i=0,\cdots,r-1$. Then, minimum $\beta$ is obtained as $\beta^{*} = \frac{\alpha}{d-r+1}$. When node repairs are performed by downloading $\beta^{*}$ symbols from $d$ nodes for each failed node, the bound in (\ref{eq:bw_upper}) reduces to
\begin{align}
\label{eq:bw_upper_locally1}
\mathcal{M} &\leq \mu r\alpha  + \min\{h,r\}\alpha
\end{align}
where $\mu =   \floorb{\frac{n-d_{\min}+1}{r+\delta-1}}$ and $h = n-d_{\min}+1 - (r +\delta -1)\floorb{\frac{n-d_{\min}+1}{r+\delta-1}}$, as defined in Theorem~\ref{thm:bw_locally}.
This establishes the file size bound for bandwidth efficient $d_{\min}$-optimal LRCs.



\subsection{Construction of repair bandwidth efficient $d_{\min}$-optimal LRCs}
\label{subsec:local_bw_cons}

Now  it is clear that a node repair within a local group of an $(r,\delta, \alpha)$  LRC is performed by treating each local group as an $(r+\delta-1,r\alpha, \delta, \alpha, \beta^{*}, d)$ MSR regenerating code. 
Remarkably, the code obtained by Construction I described in Section~\ref{subsec:optimal_local_repairable} achieves the bound (\ref{eq:bw_upper_locally1}) when an MSR code is employed for the second stage of encoding and we have $\alpha|\cM$. We establish this claim in the following theorem.

\begin{theorem}
Let $\Cc^{\rm loc}$ be a code obtained from the Construction I described in Sec.~\ref{subsec:optimal_local_repairable} with $\alpha = \frac{\cM}{k}$ and an MSR code (over $\F_q$) employed in the second stage of encoding to generate local parities. If $\alpha|\cM$, then $\Cc^{\rm loc}$ attains the bound (\ref{eq:bw_upper_locally1}), i.e., the size of a file that can be stored by using $\Cc^{\rm loc}$ satisfies
\[\mathcal{M} = \mu r\alpha  + \min\{h,r\}\alpha
\]
where $\mu = \floorb{\frac{n-d_{\min}+1}{r+\delta-1}}$ and $h = n-d_{\min}+1 - (r +\delta -1)\floorb{\frac{n-d_{\min}+1}{r+\delta-1}}$.
\end{theorem}

\begin{proof} Lets assume that $\alpha|\cM$. Then, we can write $\cM=\alpha(\alpha_1r+\beta_1)$, for some integers $0\leq \alpha_1, \beta_1$, s.t. $\beta_1\leq r-1$. Then, by~(\ref{eq:upp_bound}),
\begin{itemize}
  \item If $\beta_1>0$ then
  $n-d_{\min}+1=(\alpha_1r+\beta_1)+\alpha_1(\delta-1)=(r+\delta-1)\alpha_1+\beta_1$,
  hence $h=\beta_1<r$ and
  $\floorb{\frac{n-d_{min}+1}{r+\delta-1}}r\alpha  + \min\{h,r\}\alpha=\alpha_1 r\alpha+\beta_1\alpha=\cM$.
  \item If $\beta_1=0$ then $n-d_{\min}+1=\alpha_1r+(\alpha_1-1)(\delta-1)=(r+\delta-1)(\alpha_1-1)+r$, hence $h=r$ and
  $\floorb{\frac{n-d_{min}+1}{r+\delta-1}}r\alpha  + \min\{h,r\}\alpha=(\alpha_1-1) r\alpha+r\alpha=\cM$.
\end{itemize}

This establishes that the output of Construction I $\Cc^{\rm loc}$ attains the file size bound given in (\ref{eq:bw_upper_locally1}) when an MSR code is used to generate its local parities.
\end{proof}

In the following example, we illustrate the aforementioned construction for repair bandwidth efficient LRCs for a particular choice for system parameters.

\begin{example}
Consider the parameters given in Example~\ref{ex:array}. Now we apply an $(r+\delta-1=5,r\alpha=12,d_{\min}=3,\alpha=4,\beta=2,d = 4)$ exact-MSR code (e.g., (5,3)-zigzag code~\cite{zigzag13}) in each group instead of an MDS array code. For these parameters, $h=1$, and by (\ref{eq:bw_upper_locally1}), $\cM\leq 2\cdot3\cdot 4+1\cdot 4=28$, thus the code attains the bound (\ref{eq:bw_upper_locally1}). Moreover, each failed node can be repaired bandwidth efficiently as an exact-MSR code is used within each local group.

\end{example}

\section{Secrecy in Locally Repairable DSS}
\label{sec:secrecy_local}

In this section, we analyze locally repairable DSS in the presence of security constraints. We restrict ourselves to $d_{\min}$-optimal locally repairable DSS and consider the eavesdropping model defined in Section~\ref{subsec:eavesdropper}. We first derive a generic upper bound on secrecy capacity, the amount of data that can be stored on a DSS without leaking any information to an eavesdropper, for an $(r,\delta,\alpha, \beta, d)$ LRC. We later specialize the bound to specific cases of system parameters. While addressing specific cases, we also present constructions of secure coding schemes that achieve the respective upper bounds for certain set of  system parameters.

Similar to Section~\ref{subsec:local_repair_bw}, we associate a data collector with $n-d_{\min}+1$ nodes it contacts to reconstruct the original file. Let $\Kc_i$ denote the indices of nodes that are contacted by the data collector in $i$-th local group, and $\Kc = \cup_{i=1}^{g}\Kc_i$ with $|\Kc| = n-d_{\min}+1$. Again, we classify eavesdropped nodes into two classes: $\Ec_1$ contains storage-eavesdropped nodes ($\ell_1$ nodes in total), and $\Ec_2$ contains download-eavesdropped nodes ($\ell_2$ nodes in total). Considering $i$-th local group, we denote the set of indices of storage-eavesdropped nodes in $i$th local group as $\Ec_1^i$ and download-eavesdropped nodes as $\Ec_2^i$. Here, we have $\Ec_1 = \cup_{i=1}^{g}\Ec^i_1$, $\Ec_2 = \cup_{i=1}^{g}\Ec^i_2$. Let $\ell_1^i=|\Ec^i_1|$ and $\ell_2^i=|\Ec^i_2|$ denote the number of storage-eavesdropped and download-eavesdropped nodes in $i$-th local group, respectively. Note that $\sum_{i=1}^{g} l_1^i = \ell_1$ and $\sum_{i=1}^{g}l_2^i = \ell_2$. We denote $\mathcal{X}$ to represent set of tuples $\left(\{\Ec^i_1\}_{i=1}^{g},\{\Ec^i_2\}_{i=1}^{g},\{\Kc_i\}_{i=1}^{g}\right)$ satisfying these requirements. In the following, we provide our generic upper bound on the secrecy capacity of $(r,\delta, \alpha, d, \beta)$ LRCs against an $(\ell_1,\ell_2)$-eavesdropper.

\begin{lemma}
\label{lem:local_secrecy_bound}
For a DSS employing an $(r, \delta, \alpha, \beta, d)$ LRC that is secure against an $(\ell_1,\ell_2)$-eavesdropper, we have
\begin{align}
\label{eq:local_secrecy_bound}
\mathcal{M}^{s} \leq \sum_{i = 1}^{g}H(\sv_{\Kc_i}|\sv_{\Ec^i_1}, \dv_{\Ec^i_2})~\forall~{\left(\{\Ec^i_1, \Ec^i_2, \Kc_i\}_{i=1}^{g}\right) \in \mathcal{X}},
\end{align}
where $\Xc$ denotes the set of tuples $(\{\Ec^i_1, \Ec^i_2, \Kc_i\}_{i=1}^{g})$ that are allowed under our model.
\end{lemma}


\begin{proof}
Without loss of generality,  we can focus on sets of indices $\{\Ec^i_1\}_{i = 1}^{g}$ and $\{\Ec^i_2\}_{i = 1}^{g}$ such that $|\Ec^i_1 \cup \Ec^i_2| \leq r$ for the purpose of obtaining an upper bound on the secrecy capacity as eavesdropping $r$ nodes in a group gives the eavesdropper all the information that the group has to offer. As introduced in Section~\ref{sec:sys_model}, we represent stored and downloaded content at node $i$ (set $\Ac$) as $\xv_i$ and $\dv_i$ (respectively, $\xv_{\Ac}$ and $\dv_{\Ac}$). We assume that $(\Kc_1,\ldots, \Kc_g)$ such that $\Ec^i_1 \cup \Ec^i_2 \subseteq \Kc_i$ or $\Kc_i = \emptyset$. Note that we need that $|\Ec_1| + |\Ec_2| = \ell_1 + \ell_2 < \ceilb{\frac{\Mc}{\alpha}}$ in order to have a non-zero secure file size.

\begin{align}
H(\fv^{s}) &= H(\fv^{s}|\xv_{\Ec_1}, \dv_{\Ec_2}) \label{eq:sec1}& \\
&= H(\fv^{s}|\xv_{\Ec_1}, \dv_{\Ec_2}) - H(\fv^{s}|\xv_{\Ec_1}, \dv_{\Ec_2}, \xv_{\Kc})  \label{eq:rec1} &\\
& =  I(\fv^{s}; \xv_{\Kc}|\xv_{\Ec_{1}}, \dv_{\Ec_{2}}) \n &\\
&\leq H(\xv_{\Kc}|\xv_{\Ec_1}, \dv_{\Ec_2})\n & \\
& = H(\xv_{\Kc_1},\ldots, \xv_{\Kc_g}|\xv_{\Ec^1_1},\ldots, \xv_{\Ec^g_1}, \dv_{\Ec^1_2},\ldots, \dv_{\Ec^g_2}) &\n\\
& \leq \sum_{i = 1}^{g}H(\xv_{\Kc_i}|\xv_{\Ec^i_1}, \dv_{\Ec^i_2}), &\label{eq:upp_break}
\end{align}
where \eqref{eq:sec1} follows from the secrecy constraint, i.e., $I(\fv^{s};\xv_{\Ec_1}, \dv_{\Ec_2}) = H(\fv^{s}) - H(\fv^{s}|\xv_{\Ec_1}, \dv_{\Ec_2}) = 0$, and \eqref{eq:rec1} follows by the ability of the data collector to obtain the original data from $\xv_{\Kc}$. Since we get one such upper bound for each choice of $(\{\Ec^i_1\}_{i=1}^{g},\{\Ec^i_2\}_{i=1}^{g},\{\Kc_i\}_{i=1}^{g})$, we have
\begin{align}
\Mc^{s} = H(\fv^{s}) &\leq \sum_{i = 1}^{g}H(\xv_{\Kc_i}|\xv_{\Ec^i_1}, \dv_{\Ec^i_2})~\forall~{\left(\{\Ec^i_1, \Ec^i_2, \Kc_i\}_{i=1}^{g}\right) \in \mathcal{X}}.\n
\end{align}
\end{proof}
Now we consider two cases depending on the number of local parities per local group: 1) single parity node per local group ($\delta = 2$) and 2) multiple parity nodes per local group ($\delta > 2$). The analysis of the first case, for single parity node per local group, shows that the secrecy capacity of such coding schemes degrade substantially in the presence of an eavesdropper that can observe the data downloaded during node repairs. The second case, multiple parity nodes per local group, allows the node repair to be performed with smaller repair bandwidth, which results in lower leakage to eavesdroppers observing downloaded data. In both cases, we use the vectors $\lv_{1} = (l^1_1,\ldots, l^g_1)$ and $\lv_{2} = (l^1_2,\ldots, l^g_2)$ to represent a pattern of eavesdropped nodes. In what follows, we use $\mu$ and $h$ as short-hand notation for $\floorb{\frac{n-d_{\min}+1}{r+\delta-1}}$ and $n-d_{\min} + 1 - (r+ \delta -1)\floorb{\frac{n-d_{min}+1}{r + \delta -1}}$, respectively.

\subsection{Case 1: Single local parity per local group $(\delta = 2)$}
\label{subsec:single_local}
For an LRC with single parity node per local group, a newcomer node downloads all the data stored on other nodes in the local group it belongs to during a node repair. Since the data stored on each node is a function of the data stored on any set of $r$ nodes in the local group it belongs to, all the information in the local group is revealed to an eavesdropper that observes a single node repair in the local group. In other words, we have $H(\sv_{\Gc_i}|\dv_{\Ec_2^i}) = 0 \Rightarrow H(\sv_{\Kc_i}|\dv_{\Ec_2^i}) = 0$, whenever $\Ec_2^i \neq \emptyset$. We use this fact to present the following result on the secrecy capacity of an LRC with $\delta = 2$.

\begin{theorem}
\label{thm:secrecy_capacity_single}
The secrecy capacity of an $(r, \delta = 2, \alpha, \beta = \alpha, d = r)$ LRC against an $(\ell_1, \ell_2)$-eavesdropper is
\begin{align}
\label{eq:upp_delta_2}
\Mc^{s} = \left[\mu r + h - (\ell_2r +\ell_1)\right]^{+}\alpha,
\end{align}
where $[a]^{+}$ denotes $\max\{a,0\}$.
\end{theorem}
\begin{proof}
First, we use Lemma~\ref{lem:local_secrecy_bound} to show that RHS of \eqref{eq:upp_delta_2} serves as an upper bound on the secrecy capacity of an LRC with $\delta = 2$. Consider a data collector with $\Kc_1 = \Gc_1, \Kc_2 = \Gc_2,\ldots, \Kc_{\mu} = \Gc_{\mu}, \Kc_{\mu+2} = \ldots = \Kc_{g} = \emptyset, \Kc_{\mu+1} \subset \Gc_{\mu+1}$ such that $\left|\Kc_{\mu+1}\right| = h$; and an eavesdropper with eavesdropping pattern $\lv_{2} = (1, 1, \ldots, 1, 0, \ldots, 0)$ with ones at first $\ell_2$ positions and $\lv_{1} = (0,\ldots, 0, l_1^{\ell_2+1}, \ldots, l_1^{g})$ with zeros in first $\ell_2$ positions. Note that the choice of data collector is the same as that used in the proof of Theorem~\ref{thm:bw_locally}. (See Fig.~\ref{fig:lrc_single_bound}.)

\begin{figure*}[t]
 \centering
 \includegraphics[width=0.9\textwidth]{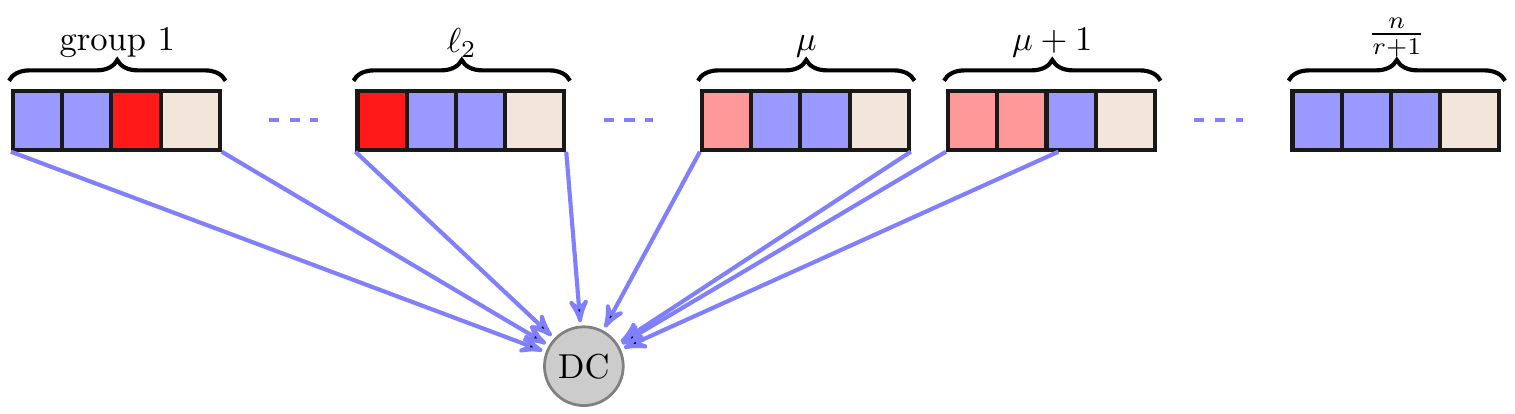}
 \caption{Illustration of the data collector and eavesdropped nodes considered to obtain the upper bound in \eqref{eq:upp_delta_2}. A configuration of eavesdropped node is highlighted using red color.  Dark red nodes and light red nodes correspond to the nodes in $\Ec_2$ and $\Ec_1$, respectively. For the purpose of obtaining the desired upper bound, we consider the scenario with one eavesdropped node repair (dark red colored node) in each of the first $\ell_2$ local groups.} \label{fig:lrc_single_bound}
\end{figure*}

In order to complete the proof, we present a secure coding scheme that shows tightness of the upper bound stated above by allowing to store a $\left[\mu r + h - (\ell_2r +\ell_1)\right]^{+}\alpha$ symbols (over $\F_{q^m}$) long file $\fv^{s}$ on a locally repairable DSS without leaking any information to an $(\ell_1, \ell_2)$-eavesdropper. Without loss of generality we assume that $\mu r + h - (\ell_2r +\ell_1) > 0$.
\begin{enumerate}
\item Append $(\ell_2r +\ell_1)\alpha$ random symbols (independent of file $\fv^{s}$) over $\mathbb{F}_{q^m}$, $\mathbf{r} = (r_1, \ldots, r_{(\ell_2r +\ell_1)\alpha})$, with $\left(\mu r + h - (\ell_2r +\ell_1)\right)\alpha$ symbols of file $\fv^s = (f_1,\ldots, f_{\left(\mu r + h - (\ell_2r +\ell_1)\right)\alpha})$.
\item Encode these $\Mc = \left(\mu r + h \right)\alpha$ symbols (including both $\rv$ and $\fv^{s}$) using an $[\Mc, \Mc, 1]_{q^m}$ Gabidulin code following encoding process specified in Section~\ref{subsec:gabidulin}.
\item Encode $\Mc$ symbols of the Gabidulin codeword using a $d_{\min}$-optimal $(r, \delta =2, \alpha)$ LRC, e.g., coding scheme proposed in \cite{DimDim12} or the construction I presented in Section~\ref{subsec:optimal_local_repairable} with $\delta = 2$.
\end{enumerate}

To prove secrecy of the proposed scheme against an $(\ell_1, \ell_2)$-eavesdropper, we use Lemma~\ref{thm:SecrecyLemma}.  It follows from Lemma~\ref{thm:SecrecyLemma} that the file is secured against an $(\ell_1,\ell_2)$-eavesdropper if (i) $H(\mathbf{e}) \leq H(\mathbf{r})$ (which is trivially true as the eavesdropper observes at most $(\ell_2r +\ell_1)\alpha$ linearly independent symbols) and (ii) $H(\rv|\fv^{s},\ev) = 0$. It remains to show the latter requirement also holds.

Given $\ev$, the eavesdropper knows the evaluations of a linearized polynomial at $|\ev| = (\ell_2r + \ell_1)\alpha$ linearly independent (over $\F_q$) points in $\F_{q^m}$. This observation follows from the fact that local code within each group is an MDS array code, which allows us to use Lemma~\ref{lem:linearized_propertyI}. When eavesdropper is also given access to $\fv^{s}$, she can remove the contribution of $\fv^{s}$ from the known $(\ell_2r + \ell_1)\alpha$ evaluations of the linearized polynomial
$$
\mathpzc{f}(y) = \sum_{i = 1}^{(\ell_2r + \ell_1)\alpha}r_iy^{q^{i-1}} + \sum_{i = (\ell_2r + \ell_1)\alpha+ 1}^{\left(\mu r + h\right)\alpha}f_{i - (\ell_2r + \ell_1)\alpha}y^{q^{i-1}}
$$
to obtain $(\ell_2r + \ell_1)\alpha$ evaluations $\tilde{\ev}$ of another linearized polynomial
$$
\tilde{\mathpzc{f}}(y) = \sum_{i = 1}^{(\ell_2r + \ell_1)\alpha}r_iy^{q^{i-1}}
$$
at  $(\ell_2r + \ell_1)\alpha$  linearly independent (over $\F_q$) points in $\F_{q^m}$. From these evaluations of $\tilde{\mathpzc{f}}(\cdot)$, using Remark~\ref{rem:gabidulin_rec}, the eavesdropper can decode the random symbols $\rv$, which implies that $H(\rv|\fv^{s},\ev) = 0$. This completes the proof of secrecy of the proposed coding scheme.
\end{proof}
\begin{remark}
It follows from Theorem~\ref{thm:secrecy_capacity_single} that every additional node repair that an eavesdropper is allowed to observe reduces the secrecy capacity of locally repairable DSS with single local parity by $r\alpha$. Moreover, DSS can not secure any data without leaking useful information to the eavesdropper for $\ell_2 > \mu$.
\end{remark}

\subsection{Case 2: Multiple local parity nodes per local group ($\delta >  2$)}
\label{subsec:secure_LRC_2}

For LRCs with $\delta >2$ that allow only na\"{\i}ve node repair, i.e., a newcomer downloads all the information from $r$ out of $r + \delta - 2$ surviving nodes in its local group, we again have $H(\sv_{\Gc_i}|\dv_{\Ec_2^i}) = 0 \Rightarrow H(\sv_{\Kc_i}|\sv_{\Ec^i_1}, \dv_{\Ec_2^i}) \leq H(\sv_{\Kc_i}|\dv_{\Ec_2^i}) = 0$,  whenever $|\Ec_2^i| > 0$. Therefore, the characterization of  the secrecy capacity for such LRCs is similar to the previous case, which we omit here for brevity. The main aim in the present section is to show that using regenerating codes within local groups (when $\delta>2$) can improve the secrecy capacity of DSS against an $(\ell_1, \ell_2)$-eavesdropper. Here, we focus on MSR-LRCs, the LRCs that give an MSR code when they are restricted to any local group. 


In the following, we assume that node repairs are performed with a newcomer node downloading $\beta$ symbols from each of $d$ surviving nodes it contacts in its own local group. Moreover, we assume that local node repairs are performed in a bandwidth efficient manner, i.e., $\beta = \frac{\alpha}{d - r + 1}$. Here, we remind again that  $\mu$ and $h$ denote $\floorb{\frac{n-d_{\min}+1}{r+\delta-1}}$ and $n-d_{\min} + 1 - (r+ \delta -1)\floorb{\frac{n-d_{min}+1}{r + \delta -1}}$, respectively.

\begin{proposition}
\label{prop:sec_lrc_case_2_bound}
For an $(r,\delta > 2, \alpha, \beta, d)$ MSR-LRC, the secrecy capacity against $(\ell_1, \ell_2)$-eavesdropping model satisfies
\begin{align}
\label{eq:upp_msr_case2_final}
\Mc^{s} &\leq  \left(\rho(r-\xi -1) - \hat{\xi}\right)\left(\alpha - \theta(\alpha, \beta, \xi+1)\right) \nonumber \\
&~~+ \left((\mu - \rho)(r- \xi) - \tilde{\xi}\right)\left(\alpha - \theta(\alpha, \beta, \xi)\right)\nonumber \\
&~~+\left(\min\{r,h\} - \nu - \tilde{\nu}\right)\left(\alpha - \theta(\alpha, \beta, \nu)\right),
\end{align}
where  $\xi$,~$\rho$,~and~$\nu$ are positive integers such that $\nu \leq \min\{h, \xi\}$, and $\ell_2 = \xi\mu + \rho + \nu$; and $\tilde{\nu}$,~$\hat{\xi}$,~and~$\tilde{\xi}$ are positive integers such that
\begin{align}
\tilde{\nu} &= \min\{(\min\{h, r\} - \nu), \ell_1\}, \n \\
\tilde{\xi} &= \min\{(\mu - \rho)(r - \xi), [\ell_1 - \tilde{\nu}]^{+}\}, \\
\hat{\xi} &=  \min\{\rho(r - \xi - 1), [\ell_1 - \tilde{\nu} - \tilde{\xi}]^{+}\}.
\end{align}
 Here, $\theta(\alpha, \beta, t)$ denotes the amount of information that an eavesdropper receives from one intact node (a node not eavesdropped) during the repair of  $t$ eavesdropped nodes in its local group.
\end{proposition}

\begin{proof}
First, we prove the upper bound on the secrecy capacity for an MSR-LRC as specified in \eqref{eq:upp_msr_case2_final}. For an MSR-LRC, we can apply Theorem~\ref{thm:sec_cap} within $i$-th local group to obtain
\begin{align}
\label{eq:upp_msr_case2}
H(\sv_{\Kc_i}|\sv_{\Ec^i_1}, \dv_{\Ec^i_2}) \leq \sum_{j=1}^{\min\{{|\Kc_i|,r}\}-l^i_1 - l^i_2}\left(\alpha - \theta\left(\alpha, \beta,l_2^i\right)\right)& \n \\
= \left(\min\{{|\Kc_i|,r}\}-l^i_1 - l^i_2\right)\left(\alpha - \theta\left(\alpha, \beta,l_2^i\right)\right).&
\end{align}

Now, we consider the data collector associated with the pattern $(\Kc_1,\ldots, \Kc_g)$ used in Section~\ref{subsec:single_local}, i.e., $\Kc_1 = \Gc_1, \Kc_2 = \Gc_2,\ldots, \Kc_{\mu} = \Gc_{\mu}$;~ $\Kc_{\mu+2} = \cdots = \Kc_{g} = \emptyset$; and $\Kc_{\mu+1} \subset \Gc_{\mu+1}$ such that $\left|\Kc_{\mu+1}\right| = h$; and an eavesdropper with eavesdropping pattern $\lv_{2}$ such that $l_2^1 = \cdots = l_2^{\rho} = \xi + 1$;~$l_2^{\rho+1} = \cdots =  l_2^{\mu} =  \xi$;~$l_2^{\mu + 2} = \cdots = l_2^g = 0$; and~$l_2^{\mu + 1} = \nu$. Given this particular choice of a DC and a pattern for download-eavesdropped nodes $\lv_{2}$, using Lemma~\ref{lem:local_secrecy_bound} and (\ref{eq:upp_msr_case2}), we get an upper bound for each pattern of storage-eavesdropped nodes $\lv_{1} = (l_1^1,\ldots, l_1^g)$ as follows:
\begin{align}
\label{eq:upp_msr_case3}
\Mc^{s} &\leq \sum_{i=1}^{\rho}(r-(l_1^i+ \xi+1))\left(\alpha - \theta(\alpha, \beta, \xi+1)\right) \nonumber \\
&~~~~~+ \sum_{i=\rho+1}^{\mu}(r-(l_1^i+ \xi))\left(\alpha - \theta(\alpha, \beta, \xi)\right)\nonumber \\
&~~~~~+(\min\{r,h\} - (l_1^{\mu+1}+\nu))\left(\alpha - \theta(\alpha, \beta, \nu)\right).
\end{align}
Now, we choose a specific pattern of $\ell_1$ storage-eavesdropped nodes such that
\begin{align}
l^{\mu + 2}_1 = \cdots = l^g_1 &= 0, \n \\
l^{\mu+1}_1 = \tilde{\nu} = &\min\{(\min\{h, r\} - \nu), \ell_1\}, \n \\
\sum_{i = \rho+1}^{\mu}l^{i}_1 = \tilde{\xi} = &\min\{(\mu - \rho)(r - \xi), [\ell_1 - \tilde{\nu}]^{+}\}, \n \\
\sum_{i = 1}^{\rho}l^{i}_1 = \hat{\xi} =  &\min\{\rho(r - \xi - 1), [\ell_1 - \tilde{\nu} - \tilde{\xi}]^{+}\}.
\end{align}
Substituting this choice of $\{l^i_1\}_{i = 1}^{g}$ in \eqref{eq:upp_msr_case3}, we obtain the bound in \eqref{eq:upp_msr_case2_final}. Note that we choose the pattern for storage-eavesdropped nodes such that $\Ec_1 \subset \Kc$.
\end{proof}

\begin{remark}
\label{remark:goparaju_lrc}
Restricting ourselves to MSR-LRCs with $d = r + \delta - 2$, it follows from Lemma~\ref{lem:intersection} that for $t \leq 2$ we have the following.
\begin{equation}
\label{eq:upp_step4}
\theta\left(\alpha, \beta, t\right) \geq \left\{\begin{array}{cc}
                \beta, & \textmd{ if } t =1 \\
                2 \beta-\frac{\alpha}{(\delta-1)^2}, & \textmd{ if } t = 2
              \end{array}\right.
\end{equation}

Here we also note that, using the upper bound on the secrecy capacity at the MSR point from \cite{GRCP13}, when $\Ec_2  \subset [k]$ and $d = r + \delta - 2$, we can use the following in the upper bound in Proposition~\ref{prop:sec_lrc_case_2_bound}.
$$
\theta\left(\alpha, \beta, t\right)  \geq \alpha - \left(1 - \frac{1}{\delta -1}\right)^{t}\alpha.
$$
\end{remark}

%

Now, we present a construction for a secure MSR-LRC with $d = r + \delta - 2$ against $(\ell_1, \ell_2)$ eavesdropping model. We then establish the claim of security for the construction. Since we are interested in LRCs with efficient repair bandwidth, similar to Section~\ref{sec:local_repair_bw}, we restrict ourselves to $(r +\delta - 1) | n$ case. The encoding process of the proposed coding scheme is as follows:
\begin{itemize}
\item Take $\kappa(\ell_1, \ell_2) = (\ell_2 + \ell_1)\alpha + \ell_2(r - 1)\beta$ i.i.d. random symbols $\rv$ that are distributed uniformly at random in $\F_{q^m}$,  and append file $\fv^s = (f^{s}_1,\ldots, f^{s}_{\Mc - \kappa(\ell_1, \ell_2)})$ of size $\Mc - \kappa(\ell_1, \ell_2) = (\mu r + \min\{h, r\})\alpha - \kappa(\ell_1, \ell_2)$ (over $\mathbb{F}_{q^m}$) to obtain $\fv = (\rv, \fv^{s}) = (f_1, f_2,\ldots, f_{(\mu r + \min\{h, r\})\alpha})$.
\item Encode $\cM = (\mu r + \min\{h, r\})\alpha$ symbols long $\fv$ using two stage encoding scheme of the construction for MSR-LRCs presented in Section~\ref{subsec:local_bw_cons}:
\begin{enumerate}
\item Encode $\fv$ to a codeword of an $[N = gr\alpha, K = \Mc, D = gr\alpha - \Mc+1]$ Gabidulin code over $\mathbb{F}_{q^m}$. (See Section~\ref{subsec:gabidulin} for encoding process of a Gabidulin code.)
\item Partition $gr\alpha$ symbols of the Gabidulin codeword into $g$ disjoint groups of size $r\alpha$ each. Then, apply an $(r+\delta-1, r)$ zigzag code (over $\F_q$) with $\alpha = (\delta-1)^{r}$ inside each local group at the second stage of encoding process.
\end{enumerate}
\end{itemize}


In the following result, we show that an LRC obtained by the construction proposed here, by combining secrecy pre-coding with  MSR-LRC from Section~\ref{subsec:local_bw_cons}, is secure against an $(\ell_1, \ell_2)$-eavesdropper.
\begin{proposition}
\label{prop:sec_lrc_case_2_construction}
Let $\cC$ be an MSR-LRC obtained from the construction described above. Then $\cC$ is secure against $(\ell_1, \ell_2)$ eavesdropping model. In particular, if $\ev$ denote observations of any eavesdropper belonging to $(\ell_1, \ell_2)$ eavesdropping model under consideration, then we have $I(\fv^{s};\ev) = 0$.
\end{proposition}
\begin{proof}
In order to establish the secrecy claim of the proposed scheme, we apply Lemma~\ref{thm:SecrecyLemma}. Each node repair that an eavesdropper observes provides her with $\alpha + (r - 1)\beta$ independent symbols. Therefore, an eavesdropper observing at most $\ell_2$ node repairs can observe at most $\ell_2\alpha + \ell_2(r - 1)\beta$ independent symbols; combining these with $\ell_1\alpha$ additional symbols obtained from content stored on $\ell_1$ storage-eavesdropped nodes, the eavesdropper has at most $|\ev| = (\ell_1 + \ell_2)\alpha + \ell_2(r - 1)\beta$ independent symbols. This gives us that $H(\ev) \leq H(\rv) = (\ell_1 + \ell_2)\alpha + \ell_2(r - 1)\beta$, the first requirement to apply Lemma~\ref{thm:SecrecyLemma}.

Now, we also need to show that $H(\ev|\rv,\fv^{s}) = 0$, i.e., given $\fv^{s}$ in addition to $\ev$, the eavesdropper can decode the random symbols $\rv$. The proof of this part is similar to that in Theorem~\ref{thm:secrecy_capacity_single}, and follows from $\F_q$-linearity of linearized polynomials (used in the first stage of encoding of the proposed coding scheme which utilizes Gabidulin codes) and the observation that in the worst case an eavesdropper observes $\kappa(\ell_1, \ell_2) =  (\ell_1 + \ell_2)\alpha + \ell_2(r - 1)\beta$ linearly independent symbols. We skip the remaining proof for brevity as it involves the repetition of ideas previously used in this paper. (Essentially, it follows that given access to $\fv^{s}$ and $\ev$, an eavesdropper can remove the contribution of $\fv^{s}$ from $\ev$ to obtain $\tilde{\ev}$ and decode $\rv$ using polynomial interpolation for a suitable linearized polynomial.) This completes the proof of security of the proposed coding scheme.
\end{proof}

Combining Proposition~\ref{prop:sec_lrc_case_2_bound} and Proposition~\ref{prop:sec_lrc_case_2_construction}, we obtain the following result on the secrecy capacity of LRCs with multiple local parities per local group.
\begin{theorem}
\label{thm:sec_lrc_case_2_final}
For an $(r, \delta, \alpha, \beta, d)$ MSR-LRC with multiple local parities per local group, the secrecy capacity against an $(\ell_1, \ell_2)$-eavesdropper with $\ell_2 r + \ell_1 \leq \mu r + \min\{h, r\}$ is the following.
\begin{align}
\Mc^{s} = (\mu r + \min\{h, r\} - \ell_2 - \ell_1)\alpha - \ell_2(r - 1)\beta.
\end{align}
\end{theorem}
\begin{proof}
Given the restriction that $\ell_2r + \ell_1 \leq  \mu r + \min\{h, r\}$, we have $\xi = \nu = \hat{\xi} =  0$, $\rho = \ell_2$, and $\tilde{\xi} + \tilde{\nu} = \ell_1$ as a valid choice of various parameters in Proposition~\ref{prop:sec_lrc_case_2_bound}.

With these parameters the bound in \eqref{eq:upp_msr_case2_final} reduces to
$$
\Mc^{s} \leq (\mu r + \min\{h, r\} - \ell_2 - \ell_1)\alpha - \ell_2(r - 1)\beta.
$$

Proposition~\ref{prop:sec_lrc_case_2_construction} implies that this upper bound is achievable using the secure MSR-LRC proposed in this paper.
\end{proof}

\section{Conclusion}
\label{sec:con}

Distributed storage systems store data on a network of nodes. Due to their distributed nature, these systems may have to satisfy security and local-repairability constraints apart from ensuring resilience against node failures. Regenerating codes proposed for DSS address the node failure resilience while efficiently trading off storage for repair bandwidth. In this paper, we considered security and local-repairability aspects of coding schemes for DSS. The eavesdropping model analyzed in this paper belongs to the class of passive attack models, where the eavesdroppers observe the content of the nodes in the system. Accordingly, we considered an $(\ell_1,\ell_2)$-eavesdropper, where the content of any $\ell_1$ nodes and the downloaded information for any $\ell_2$ nodes are leaked to the eavesdropper. With such an eavesdropping model, we first focused on the classical setup which is resilient against single node failure at a time (without small locality constraints). Noting that the secrecy capacity of this setting is open at the minimum storage regenerating point, we provided upper bounds on the secure file size and established the secrecy capacity for any $(\ell_1,\ell_2)$ with $\ell_2\leq 2$ (whenever the eavesdropper can observe repairs of nodes in a particular set of $k$ nodes). Our coding scheme achieving this bound also provides a better rate compared to the existing schemes.
Then, we shifted focus on local-repairability constraint, and studied the general scenario of having multiple parity nodes per local group. For this setting, we derived a new upper bound on minimum distance of locally repairable codes. We also presented a coding scheme that is optimal with respect to the derived bound. Similar to the trade-off analysis for the classical setup, we then studied the bandwidth efficient LRCs, where we proposed a new bound and a coding scheme which is both $d_{\min}$-optimal and repair bandwidth efficient. This bandwidth efficient locally repairable setting is also analyzed under security constraints. We obtained a secure file size upper bound on such codes. We then presented  secure codes achieving the file size bound, and hence established the secrecy capacity, under special cases.

Here, we point out that another family of bandwidth efficient LRCs (MBR-LRCs) where restriction of an LRC to each local group is an MBR code has also been considered in \cite{KPLV12}.  In a more recent work~\cite{LocalMBR}, Kamath et al. present a different construction for MBR-LRC which is similar to the construction for MSR-LRC presented in Section~\ref{subsec:local_bw_cons}. On the similar note, following the mechanism given in~\cite{KPLV12}, scalar LRCs presented in this paper can also be used to design novel MBR-LRCs.

We list some avenues for further research here. 1) We first note that the novel bound on the secrecy capacity that we establish for the minimum storage point allows for counting part of the data downloaded as additional leakage, and hence provide a tighter bound than the existing ones. Recently, a new upper bound on secrecy capacity against an $(\ell_1, \ell_2)$-eavesdropper at the MSR point is presented in \cite{GRCP13}. The bound follows by substituting a tighter lower bound on the information leaked from an intact node to an eavesdropper during repair of $\ell_2$ nodes
~in Theorem~\ref{thm:sec_cap} in Section~\ref{sec:BW_improved}. Moreover, the coding scheme described in Section~\ref{subsec:secure_msr} also establish the tightness of the upper bound presented in \cite{GRCP13}; as a result, settling the problem of securing data against an $(\ell_1, \ell_2)$-eavesdropper at the MSR point when eavesdropped node repairs are restricted to a subset of $k$ systematic nodes. 
Yet, the secrecy capacity for eavesdroppers that can observe the parity nodes in addition to the systematic nodes remain as an open problem. Thus, new codes or improved bounds are of definite interest for secure MSR codes. 2) For LRCs, we utilized MRD coding as the secrecy pre-coding, which requires extended field sizes. Designing codes that achieve the stated bounds with lower field sizes is an interesting problem. Also a construction of $d_{min}$-optimal LRC with the parameters which are not covered by Theorem~\ref{thm:parameters} still remains an open problem.
3) One can also consider cooperative (or, multiple simultaneous node failure) repair~\cite{Hu:Cooperative10,Kermarrec:Repairing11,Shum:Existence11} in a DSS. Secure code design in such a scenario is recently considered in~\cite{KRV12}. Codes having both cooperative and locally repairable features can be studied. As distributed systems, storage problem may exhibit simultaneous node failures that need to be recovered with local connections. According to our best knowledge, this setting has not been studied (even without security constraints). Our ongoing efforts are on the design of coding schemes for DSS satisfying these properties.

%
%


\appendices

\section{Proof of Lemma~\ref{thm:SecrecyLemma}}
\label{app:SecrecyLemma}
\begin{proof}
The proof follows from the classical techniques given
by~\cite{Wyner:The75}, where instead of
$0$-leakage, $\epsilon$-leakage rate is considered.
(The application of this technique in DSS is first
considered in~\cite{SRK_globecom11}.)
We have
\begin{eqnarray}
I(\fv^{s};\ev) &=& H(\ev)-H(\ev|\fv^{s}) \\
&\stackrel{(a)}{\leq}&H(\ev)-H(\ev|\fv^{s}) + H(\ev|\fv^{s},\rv) \\
&\stackrel{(b)}{\leq}&H(\rv)-I(\ev;\rv|\fv^{s}) \\
&\stackrel{(c)}{=}&H(\rv|\fv^{s},\ev) \\
&\stackrel{(d)}{=}&0
\end{eqnarray}
where (a) follows by non-negativity of $H(\ev|\fv^{s},\rv)$,
(b) is the condition $H(\ev)\leq H(\rv)$,
(c) is due to $H(\rv|\fv^{s})=H(\rv)$ as $\rv$ and $\fv^{s}$ are
independent, (d) is the condition $H(\rv|\fv^{s},\ev)=0$.
\end{proof}

\section{Proof of Lemma~\ref{lem:linearized_property}}
\label{app:linearized_property}

\begin{proof}
Note that
$$
\hat{c}_j = \sum_{i = 1}^{t\alpha}\hat{g}_{ij}\mathpzc{f}(y_i),
$$
where $\hat{g}_{ij}$ for $i \in [t\alpha]$ and $j \in [\hat{n}\alpha]$ denote $(i, j)$-th entry of $\hat{G}$. Since the MDS array code is defined over $\F_q$, we have the following from the $\F_q$-linearity of $\mathpzc{f}(\cdot)$. (See Remark~\ref{rem:linearized_poly}.)
\begin{align}
\label{eq:lem_linearized1}
\hat{c}_j = \sum_{i = 1}^{t\alpha}\hat{g}_{ij}\mathpzc{f}(y_i) = \mathpzc{f}\left(\sum_{i = 1}^{t\alpha}\hat{g}_{ij}y_i\right).
\end{align}
Let $\{j_1, j_2,\ldots, j_s\}$ be the index of $s$ vector symbols of codeword $\hat{\cv}$. Given this we define set $\Sc \subset [\hat{n}\alpha]$ as follows,
$$
\Sc = \{(j_1-1)\alpha+1,\ldots, j_1\alpha,\ldots, (j_s-1)\alpha+1,\ldots, j_s\alpha\}
$$
 Now, it follows from \eqref{eq:lem_linearized1} that the observed $s$ vector symbols of codeword $\hat{\cv}$ corresponds to evaluations of $\mathpzc{f}(\cdot)$ at $s\alpha$ points $\{\tilde{y}_1,\ldots, \tilde{y}_{s\alpha}\}$ such that
\begin{align}
\label{eq:lem_linearized2}
(\tilde{y}_{1},\ldots, \tilde{y}_{s\alpha}) = \left\{\sum_{i = 1}^{t\alpha}\hat{g}_{ij}y_{i}\right\}_{j \in \Sc} = (y_1,\ldots, y_{t\alpha})\hat{G}_{\Sc},
\end{align}
where $\hat{G}_{\Sc}$ denotes the sub matrix of $\hat{G}$ obtained from columns of $\hat{G}$ with their indices in set $\Sc \subset [\hat{n}\alpha]$. For an MDS array code, we have that rank of matrix $\hat{G}_{\Sc}$ (over $\F_q$) is equal to $\min\{|\Sc| = s\alpha, t\alpha\}$. Using this fact along with \eqref{eq:lem_linearized2} gives us that $\{\tilde{y}_1,\ldots, \tilde{y}_{s\alpha}\}$ are $\min\{s\alpha, t\alpha\}$ linearly independent (over $\F_q$) points in subspace spanned by $\{y_1,\ldots, y_{t\alpha}\} \subset \F_{q^m}$.
\end{proof}


\section{Proof of Lemma~\ref{lem:intersection}}
\label{sec:lemma_intersection_proof}

\begin{proof}
We prove the Lemma for $n-k = 2$, i.e., $(k+2,k)-$DSS. The proof extends to higher number of parities in straightforward manner. Consider the following encoding matrix of the $(k+2,k)$ linear code employed by the DSS
\begin{equation}
\mathbf{G}=\left[\begin{array}{cccccc}
            I & 0 & \ldots & 0 & A^T_1 & B^T_1  \\
            0 & I & \ldots & 0  & A^T_2 & B^T_2\\
            \vdots & \vdots & \ddots & \vdots & \vdots & \vdots\\
            0 & 0 & \ldots & I & A^T_k & B^T_k\\
          \end{array}
\right].
\end{equation}
Assume that a newcomer node downloads $S_{1,j}\mathbf{x}^T_{k+1}$ and $S_{2,j}\mathbf{x}^T_{k+2}$ from the first and the second parity nodes during the repair process of $j$-th systematic node. Here $S_{1,j} = V_{k+1,j}$ and $S_{2,j} = V_{k+2,j}$ are $\frac{\alpha}{2}\times\alpha$ matrices. In order to be able to perform bandwidth efficient repair using interference alignment, $\{S_{1,j}\}_{j=1}^k$ and $\{S_{2,j}\}_{j=1}^k$ satisfy
\begin{equation}
\label{eq:IAreq1}
\text{rank}\left(\begin{array}{c}
            S_{1,j}A_i\\
            S_{2,j}B_i
          \end{array}
\right) = \frac{\alpha}{2}~~~\forall i \in [k]\backslash\{j\}
\end{equation}
and
\begin{equation}
\label{eq:IAreq2}
\text{rank}\left(\begin{array}{c}
            S_{1,j}A_j\\
            S_{2,j}B_j
          \end{array}
\right) = \alpha.
\end{equation}

Note that data downloaded from $i$-th systematic node $(i \neq j)$ for node repair is $V_{i,j}[f_{(i-1)\alpha+1},\ldots, f_{i\alpha}]^T$. Since the repair matrix of node $i$ associated to node repair of $j$-th node is $V_{i,j}$, we have
\begin{equation}
V_{i,j} = S_{1,j}A_i = S_{2,j}B_i.
\end{equation}
Note that the above relationship is among subspaces. (We use uppercase letters to represent both matrices and row spaces associated with those matrices.) Using the method of induction, we now show the main claim of Lemma~\ref{lem:intersection}. Note that this proof is modification of the proof of Lemma 10 in \cite{TamoWang_isit2012}.

\textbf{Base case $(|\mathcal{A}| = 1)$:} The statement of Lemma~\ref{lem:intersection} is true for this case as we perform a bandwidth efficient node repair, where each remaining node contributes $\frac{\alpha}{2}$ independent symbols for a single node repair.


\textbf{Inductive step:} Now we assume that the statement of Lemma~\ref{lem:intersection} is true for all sets $\mathcal{A} \subseteq [k]\backslash\{i\}$ with $|\mathcal{A}| \leq t-1$ and prove it for all sets of indices of size $t$. With out loss of generality, we prove this for $\mathcal{A} = [t]$. We know from inductive hypothesis that
\begin{equation}
\dim\left(\bigcap_{j \in [t-1]}V_{i,j}\right) = \text{rank}\left(\bigcap_{j \in [t-1]}V_{i,j}\right) \leq \frac{\alpha}{2^{t-1}},
\end{equation}

Now assume that the result is false for $\mathcal{A} = [t]$, i.e.,
\begin{eqnarray}
\label{eq:assump}
\dim\left(\bigcap_{j \in [t]}V_{i,j}\right) & = &\text{rank}\left(\bigcap_{j \in [t]}V_{i,j}\right) \nonumber \\
& = & \text{rank}\left(\bigcap_{j \in [t]}S_{1,j}A_i\right) \nonumber \\
&=& \text{rank}\left(\bigcap_{j \in [t]}S_{2,j}B_i\right) \nonumber \\
&>& \frac{\alpha}{2^{t}},
\end{eqnarray}

Since $A_i$ and $B_i$ are invertible, we have $\text{rank}\left(\bigcap_{j \in [t]}S_{1,j}A_i\right) = \text{rank}\left(\bigcap_{j \in [t]}S_{1,j}\right)$ and $ \text{rank}\left(\bigcap_{j \in [t]}S_{2,j}B_i\right) = \text{rank}\left(\bigcap_{j \in [t]}S_{2,j}\right)$. Next, consider

\begin{eqnarray}
\label{eq:subspace1}
\left(\bigcap_{j \in [t]}S_{1,j}\right)A_t &=& \left(\bigcap_{j \in [t]}S_{1,j}A_t\right) \n \\
&\subseteq& \left(\bigcap_{j \in [t-1]}S_{1,j}A_t\right) \n \\
&=& \bigcap_{j \in [t-1]}V_{i,j}.
\end{eqnarray}
Here, the above equation describe the relationship among row spaces of participating matrices. Similarly, we have the following.
\begin{equation}
\label{eq:subspace2}
\left(\bigcap_{j \in [t]}S_{2,j}\right)B_t \subseteq \bigcap_{j \in [t-1]}V_{i,j}.
\end{equation}

Moreover, it follows from (\ref{eq:assump}) and the full-rankness of $A_t$ and $B_t$ that
\begin{eqnarray}
\label{eq:subspace_size}
\dim\left(\left(\bigcap_{j \in [t]}S_{1,j}\right)A_t\right) &=& \dim\left(\left(\bigcap_{j \in [t]}S_{2,j}\right)B_t\right) \n \\
&>& \frac{\alpha}{2^{t}}
\end{eqnarray}

Thus, we have two subspaces $\left(\bigcap_{j \in [t]}S_{1,j}\right)A_t$ and $\left(\bigcap_{j \in [t]}S_{2,j}\right)B_t$ of dimension strictly greater than $\frac{\alpha}{2^t}$ (see (\ref{eq:subspace_size})), which are contained in the subspace $\bigcap_{j \in [t-1]}V_{i,j}$ of dimension at most $\frac{\alpha}{2^{t-1}}$ (see (\ref{eq:subspace1}) and (\ref{eq:subspace2})). Therefore,

\begin{eqnarray}
\label{eq:contradiction}
\left(\left(\bigcap_{j \in [t]}S_{1,j}\right)A_t\right) \bigcap \left(\left(\bigcap_{j \in [t]}S_{2,j}\right)B_t\right) \neq \{0\} \nonumber \\
\Rightarrow S_{1,t}A_t \bigcap S_{2,t}B_t \neq \{0\} \nonumber
\end{eqnarray}
which is in contradiction with (\ref{eq:IAreq2}). This implies that
\begin{equation}
\dim\left(\bigcap_{j \in [t]}V_{i,j}\right) \leq \frac{\alpha}{2^t}.
\end{equation}

\end{proof}

\section{Proof of Lemma~\ref{lm:union_zigzag}}
\label{appen:union_zigzag}

\begin{proof} First, note that repair sets $Y_{j}$ for $j \in [k]$ satisfy
$$|Y_j|=p^{k-1}$$
$$|Y_{j_1}\cap Y_{j_2}|=p^{k-2}, \textmd{ for } j_1\neq j_2,$$
and in general
$$|Y_{j_1}\cap Y_{j_2}\ldots \cap Y_{j_t}|=p^{k-t}, \textmd{ for } j_1\neq j_2 \neq \ldots \neq j_t.$$

Next, we count the number of symbols stored on systematic nodes that are revealed to an eavesdropper observing repairs of $\ell_2$ systematic nodes. Let $\Ec_2 \subset [k]$ be the set of $\ell_2$ nodes where the eavesdropper observes the data downloaded during node repair.  A node participating in all $\ell_2$ eavesdropped node repairs  send $|\cup_{j\in \Ec_2}Y_j|$ distinct symbols. From inclusion-exclusion principle, we have
\begin{align*}
 |\cup_{j\in \Ec_2}Y_j| & =  \sum_{j \in \Ec_2}|Y_j| -\sum_{j_1 \neq j_2 \in \Ec_2}|Y_{j_1}\cap Y_{j_2}| & \n \\
&~~~+\sum_{j_1 \neq j_2 \neq j_3 \in \Ec_2}|Y_{j_1}\cap Y_{j_2} \cap Y_{j_3}|\ldots & \\
&=  \ell_2\cdot p^{k-1}-\binom{\ell_2}{2}\cdot p^{k-2}+\binom{\ell_2}{3}\cdot p^{k-3}\ldots & \\
 &= \sum_{i=1}^{\ell_2}(-1)^{i-1}\binom{\ell_2}{i}p^{k-i} &\nonumber \\
&=p^k-p^{k-\ell_2}(p-1)^{\ell_2}.&
\end{align*}

The eavesdropper observes data stored on $|\Ec_1 \cap [k]| + \ell_2$ systematic nodes. Therefore, the eavesdropper gets access to
\begin{align*}
 &p^k(|\Ec_1 \cap [k]| + \ell_2)+(k-|\Ec_1 \cap [k]|-\ell_2)|\cup_{j\in \Ec_2} Y_j| & \\
 &= p^k(|\Ec_1 \cap [k]|+\ell_2) & \n \\
&~~~+(k-|\Ec_1 \cap [k]|-\ell_2)(p^k-p^{k-\ell_2}(p-1)^{\ell_2})& \\
 & = kp^k-p^k(k-|\Ec_1 \cap [k]|-\ell_2)\left(1-\frac{1}{p}\right)^{\ell_2}&
 \end{align*}
systematic symbols.

We now focus on the parity symbols that are observed by the eavesdropper. The eavesdropper gains all the symbols stored on $(\ell_1 - |\Ec_1 \cap [k]|)$ parity nodes. In addition, the eavesdropper gets $|\cup_{j\in \Ec_2}Y_j| = p^k-p^{k-\ell_2}(p-1)^{\ell_2}$ symbols from the remaining parity nodes during node repair of $\ell_2$ systematic nodes. However, $p|\cup_{j\in \Ec_2}Y_j|$ of these symbols ($|\cup_{j\in \Ec_2}Y_j|$ from each parity node) are function of the systematic symbols that eavesdropper has access to. Therefore the eavesdropper observes at most
\begin{align*}
\centering
&kp^k-p^k(k-|\Ec_1 \cap [k]|-\ell_2)\left(1-\frac{1}{p}\right)^{\ell_2} & \n \\
&~~~+ (\ell_1 -|\Ec_1 \cap [k]|)(p^k - |\cup_{j\in \Ec_2}Y_j|) \\
&= kp^k  - p^k(k-|\Ec_1 \cap [k]|-\ell_2)\left(1-\frac{1}{p}\right)^{\ell_2} & \n \\
&~~~ +p^k (\ell_1 -|\Ec_1 \cap [k]|)\left(1 - \frac{1}{p}\right)^{\ell_2} \\
&=  kp^k  - p^k(k-\ell_1-\ell_2)\left(1-\frac{1}{p}\right)^{\ell_2}
\end{align*}
linearly independent symbols.
\end{proof}


\section{Proof of Theorem~\ref{thm:parameters}}
\label{ap:appendix}

\begin{proof}
In order to prove that $\Cc^{\rm loc}$ attains the bound in~(\ref{eq:upp_bound}), it is enough to show that any pattern of $E\triangleq n - \ceilb{\frac{\mathcal{M}}{\alpha}}  - \left(\ceilb{\frac{\mathcal{M}}{r\alpha}}-1\right)(\delta - 1)$ node erasures can be corrected by $\Cc^{\rm loc}$. Towards this, we  prove that any $E$ erasures of $\Cc^{\rm loc}$  correspond to at most $D-1$ rank erasures of the underlying $[N,\cM,D]_{q^m}$ Gabidulin code $\cC^{\rm{Gab}}$, and hence can be corrected by  $\cC^{\rm{Gab}}$. Here, we point out the the worst case erasure pattern is when the erasures appear in the smallest possible number of groups and the number of erasures inside a local group is the maximum possible.

First,  given $n=\frac{N}{\alpha}+\left\lceil\frac{N}{r\alpha }\right\rceil(\delta-1)$, we can rewrite $E$ in the following way:
\begin{equation}
\label{eq:bound_rewrite}
E=\frac{N}{\alpha}-\left\lceil\frac{\cM}{\alpha}\right\rceil+\left(\left\lceil\frac{N}{r \alpha}\right\rceil-\left\lceil\frac{\cM}{r \alpha}\right\rceil+1\right)(\delta-1).
\end{equation}

Let $\alpha_1, \beta_1, \gamma_1$ be the integers such that  $\cM=\alpha(\alpha_1 r+\beta_1)+\gamma_1$, where $1\leq \alpha_1\leq g =\left\lceil\frac{n}{r+\delta-1}\right\rceil$; $0 \leq \beta_1\leq r-1$; and $0\leq \gamma_1\leq \alpha-1$. We then complete the proof by treating two cases: $(r+\delta-1)|n$ and $(r+\delta-1)\nmid n$ in the following.

\begin{enumerate}
  \item If $ (r+\delta-1)|n$, then $N=g r \alpha$ and
  \begin{equation}
  \label{eq:rank_dist1}
  D-1=N-\cM=(g-\alpha_1)r\alpha-\beta_1\alpha-\gamma_1.
  \end{equation}

    \begin{itemize}
    \item If $\gamma_1=\beta_1=0$, then $\left\lceil\frac{\cM}{\alpha}\right\rceil=\alpha_1 r$ and $\left\lceil\frac{\cM}{r \alpha}\right\rceil=\alpha_1$. In this case,~(\ref{eq:bound_rewrite}) gives $E=(g-\alpha_1)(r+\delta-1)+(\delta-1)$. Therefore, in the worst case, we have $(g-\alpha_1)$ local groups with all of their nodes erased, and one additional local group with $\delta-1$ erased nodes. This corresponds to $r \alpha$ rank erasures in $(g-\alpha_1)$ groups of the corresponding Gabidulin codeword. (See  Lemma~\ref{lem:linearized_propertyI}.) Since by~(\ref{eq:rank_dist1}), $D-1=(g-\alpha_1)r\alpha$, this erasures can be corrected by the Gabidulin code.

    \item If $\gamma_1=0$ and $\beta_1>0$, then $\left\lceil\frac{\cM}{\alpha}\right\rceil=\alpha_1 r+\beta_1$ and $\left\lceil\frac{\cM}{r \alpha}\right\rceil=\alpha_1+1$. Now, it follows from (\ref{eq:bound_rewrite}) that $E=(g-\alpha_1-1)(r+\delta-1)+(r+\delta-1-\beta_1)$. Hence, in the worst case, we have $(g-\alpha_1-1)$ local groups where all of their $(r + \delta - 1)$ nodes are erased, and one additional local group with $r+\delta-1-\beta_1$ erased nodes. From Lemma~\ref{lem:linearized_propertyI} and~(\ref{eq:rank_dist1}), there erasures corresponds to $(g-\alpha_1)r\alpha- \beta_1\alpha=D-1$ rank erasures that can be corrected by the Gabidulin code.

    \item If $\gamma_1>0$, then $\left\lceil\frac{\cM}{\alpha}\right\rceil=\alpha_1 r+\beta_1+1$ and $\left\lceil\frac{\cM}{r \alpha}\right\rceil=\alpha_1+1$.  We have $E=(g-\alpha_1-1)(r+\delta-1)+(r+\delta-1-\beta_1-1)$ in this case (see \eqref{eq:bound_rewrite}). Therefore, in the worst case, there are $(g-\alpha_1-1)$ local groups with all their nodes erased, and one additional local group with $r+\delta-1-\beta_1-1$ erased nodes, which by Lemma~\ref{lem:linearized_propertyI} and by~(\ref{eq:rank_dist1}) corresponds to $ (g-\alpha_1)r\alpha- \beta_1\alpha-\alpha<D-1$ rank erasures that can be corrected by the Gabidulin code.
  \end{itemize}

  \item If $n~({\rm mod }~ r+\delta-1)-(\delta-1)\geq\left\lceil\frac{\cM}{\alpha}\right\rceil({\rm mod }~ r)>0$, then, using the fact that $n(\textmd{mod } r+\delta-1)-(\delta-1)\equiv \beta_0$, we have
   $\beta_0\geq \beta_1>0$, $N=(g-1)r\alpha+\beta_0\alpha$, $\frac{N}{\alpha}=(g-1)r+\beta_0$,
   $\left\lceil\frac{N}{r\alpha}\right\rceil=g$ and
   \begin{equation}
  \label{eq:rank_dist2}
  D-1=(g-\alpha_1-1)r\alpha+(\beta_0-\beta_1)\alpha-\gamma_1.
  \end{equation}

  \begin{itemize}
       \item If $\gamma_1=0$, then $\left\lceil\frac{\cM}{\alpha}\right\rceil=\alpha_1 r+\beta_1$ and $\left\lceil\frac{\cM}{r \alpha}\right\rceil=\alpha_1+1$. In this case, it follows from~(\ref{eq:bound_rewrite}) that $E =(g-\alpha_1-1)(r+\delta-1)+(\beta_0-\beta_1+\delta-1)$. Hence, in the worst case, we have $(g-\alpha_1-1)$ local groups with all of their nodes erased and one additional local group with $\beta_0-\beta_1+\delta-1$ erased nodes (or $\beta_0+\delta-1$ erased nodes in the smallest group, $(g-\alpha_1-2)$ groups with all the erased nodes and one group with $r+\delta-1-\beta_1$ erased nodes). This by Lemma~\ref{lem:linearized_propertyI} and by~(\ref{eq:rank_dist2}) corresponds to $ (g-\alpha_1-1)r\alpha +(\beta_0- \beta_1)\alpha=D-1$ rank erasures that can be corrected by the Gabidulin code.

    \item If $\gamma_1>0$, then $\left\lceil\frac{M}{\alpha}\right\rceil=\alpha_1 r+\beta_1+1$ and $\left\lceil\frac{M}{r \alpha}\right\rceil=\alpha_1+1$. Next, using~(\ref{eq:bound_rewrite}) we get $E=(g-\alpha_1-1) (r+\delta-1)+(\beta_0-\beta_1-1+\delta-1)$. Therefore, we have $(g-\alpha_1-1)$ local groups with all of their nodes erased, and one additional local group with $\beta_0-\beta_1-1+\delta-1$ erased nodes (or $\beta_0+\delta-1$ erased nodes in the smallest group, $(g-\alpha_1-2)$ groups with all the erased nodes and one group with $r+\delta-1-\beta_1-1$ erased nodes) in the worst case. This by Lemma~\ref{lem:linearized_propertyI}  and by~(\ref{eq:rank_dist2}) corresponds to $(g-\alpha_1-1)r\alpha +(\beta_0- \beta_1)\alpha-\alpha<D-1$ rank erasures that can be corrected by the Gabidulin code.
  \end{itemize}
\end{enumerate}
\end{proof}

%
\bibliographystyle{IEEEtran}
\bibliography{IEEEabrv,SecureLocalDSS_revise_v2}

\end{document}